 \documentclass[final,3p,12pt]{elsarticle}

\usepackage[para,online,flushleft]{threeparttable}
\usepackage[pdf]{pstricks}

\usepackage{amssymb}

 \biboptions{sort&compress}

\usepackage{array}
\usepackage{wrapfig}
\usepackage{multirow}
\usepackage{amsmath}
\usepackage{amssymb,amsthm}
\usepackage{graphicx}
\usepackage{textcomp}
\usepackage{wrapfig}
\usepackage{color}
\usepackage{bm}
\usepackage{accents}
\usepackage{cuted}
\usepackage{caption}
\usepackage{pgfplots}
\usepackage{mathrsfs}
\usepackage{setspace,wasysym}
\usepackage{lscape}
\usepackage{algpseudocode}
\usepackage[ruled,vlined,commentsnumbered,linesnumbered]{algorithm2e}
\usetikzlibrary{patterns}

\usepackage{booktabs}
\usepackage{rotating}

\newtheorem{prop}{Proposition}

\usepackage{subcaption}
\usepackage[outline]{contour}
\contourlength{1.0pt}

\usepackage{tikz}
\usetikzlibrary{matrix} 
\usetikzlibrary{arrows} 
\usetikzlibrary{arrows.meta}
\usetikzlibrary{calc} 
\usetikzlibrary{shapes}
\usetikzlibrary{fit}
\usetikzlibrary{positioning}
\usetikzlibrary{intersections,patterns,pgfplots.fillbetween}
\usetikzlibrary{calc}
\usepackage{epstopdf}
\usepackage[normalem]{ulem}

\tikzstyle{block} = [draw,rectangle,thick,minimum height=2em,minimum width=2em]
\tikzstyle{sum} = [draw,circle,inner sep=0mm,minimum size=2mm]
\tikzstyle{connector} = [->,thick]
\tikzstyle{line} = [thick]
\tikzstyle{branch} = [circle,inner sep=0pt,minimum size=1mm,fill=black,draw=black]
\tikzstyle{guide} = []

\makeatletter
\newcommand{\algorithmfootnote}[2][\footnotesize]{%
  \let\old@algocf@finish\@algocf@finish
  \def\@algocf@finish{\old@algocf@finish
    \leavevmode\rlap{\begin{minipage}{\linewidth}
    #1#2
    \end{minipage}}%
  }%
}
\makeatother

\DeclareMathAlphabet{\mathcal}{OMS}{cmsy}{m}{n} 
\renewcommand{\epsilon}{\text{\usefont{OML}{cmr}{m}{n}\symbol{15}}}

\newcommand{\transpose}{^{\mathrm{T}}}
\newcommand{\superscript}[1]{\ensuremath{^{\textrm{#1}}}}
\newcommand{\subscript}[1]{\ensuremath{_{\textrm{#1}}}}
\newcommand{\Prob}[1]{\mathbb{P}\!\left({#1}\right)}

\newcommand{\Cm}{{\mathbf C}}

\newcommand{\cb}{c\subscript{b}}

\newcommand{\D}{\boldsymbol {\mathfrak{D}}}

\newcommand{\dd}{\mathbf{d}}

\newcommand{\Evd}{\mathcal{E}}
\newcommand{\Exp}{\mathbb{E}}
\newcommand{\errorcomponent}{\epsilon} 

\newcommand{\kpre}{k\subscript{pre}}
\newcommand{\keq}{k\subscript{eq}}
\newcommand{\kpost}{k\subscript{post}}
\newcommand{\Km}{{\mathbf K}}

\newcommand{\Le}{{\mathcal{L}}}

\newcommand{\Mm}{{\mathbf M}}
\newcommand{\M}{{ \mathcal{M}}}
\newcommand{\Mk}{{\mathcal{M}}_{{k}}}

\newcommand{\Nm}{n_{\M}}

\newcommand{\nst}{n\subscript{st}}
\newcommand{\nstk}[1]{n_{\textrm{st},{#1}}}
\newcommand{\npow}{n\subscript{pow}}
\newcommand{\nth}{n_{\thetaa}} 

\newcommand{\pdf}{\mathrm{p}}
\newcommand{\pdfis}{\mathrm{q}}

\newcommand{\Qy}{Q\subscript{y}}

\newcommand{\rr}{{\mathbf r}}

\newcommand{\rk}{r\subscript{k}}

\newcommand{\ub}{u\subscript{b}}
\newcommand{\uu}{{\mathbf u }}

\newcommand{\Vt}{V_\mathrm{tot}}

\newcommand{\x}{{\mathbf x }}

\newcommand{\y}{{\mathbf y}}

\newcommand{\z}{{\mathbf z}}

\newcommand{\Sigm}{\boldsymbol{\Sigma}}

\newcommand{\Thetaa}{\boldsymbol{\Theta}}
\newcommand{\thetaa}{\boldsymbol{\theta}}
\newcommand{\zetaeq}{\zeta\subscript{eq}}

\newcommand{\subb}{\subscript{b}}
\newcommand{\subg}{\subscript{g}}
\newcommand{\subs}{\subscript{s}}

\newcommand{\ubd}{{\dot u}\subscript{b}}
\newcommand{\rd}{r_{\mathrm{d}}}

\newcommand{\ie}{{\em i.e.}, }
\newcommand{\eg}{{\em e.g.}, }
\newcommand{\fref}[1]{Figure~\ref{#1}}

\newcommand{\Tstrut}{\rule{0pt}{2.4ex}}       
\newcommand{\Bstrut}{\rule[-1.0ex]{0pt}{0pt}} 
\newcommand{\HeaderTstrut}{\rule{0pt}{3.0ex}}       
\newcommand{\HeaderBstrut}{\rule[-1.6ex]{0pt}{0pt}} 

\makeatletter
\tikzset{
  use path for main/.code={%
    \tikz@addmode{%
      \expandafter\pgfsyssoftpath@setcurrentpath\csname tikz@intersect@path@name@#1\endcsname
    }%
  },
  use path for actions/.code={%
    \expandafter\def\expandafter\tikz@preactions\expandafter{\tikz@preactions\expandafter\let\expandafter\tikz@actions@path\csname tikz@intersect@path@name@#1\endcsname}%
  },
  use path/.style={%
    use path for main=#1,
    use path for actions=#1,
  }
}
\makeatother
\newif\ifdarkmode
\darkmodefalse
\ifdarkmode
	\definecolor{commentcolor}{rgb}{0,1,0}
	\definecolor{sourcecolor}{rgb}{1,.5,1}
	\pagecolor{black}\color{white} 
	\definecolor{newblue}{rgb}{0.7,0.7,1}
\else
	\definecolor{commentcolor}{rgb}{1,0,1}
	\definecolor{sourcecolor}{rgb}{0,.5,0}
	\definecolor{newblue}{rgb}{0,0,1}
\fi
\providecommand{\ouremptyset}{\varnothing} 
\makeatletter\newcommand{\raisemath}[1]{\mathpalette{\raisem@th{#1}}}\newcommand{\raisem@th}[3]{\raisebox{#1}{$#2#3$}}\makeatother 
\providecommand{\chisub}[1]{\chi_{\raisemath{-2pt}{#1}}} 
\providecommand{\widehatchisub}[1]{\widehat\chi_{#1}} 
\providecommand{\chii}{\chisub{i}}
\providecommand{\units}[1]{\,\mbox{#1}}

\begin{document}

\begin{frontmatter}

\title{Likelihood Level Adapted Estimation of Marginal Likelihood for Bayesian Model Selection} 



\address[label1]{Department of Mechanical Engineering, Northern Arizona University, Flagstaff, AZ 86011, USA}

\author[label1]{Subhayan De\corref{cor1}} 
\ead{Subhayan.De@nau.edu} 
\cortext[cor1]{Corresponding author} 

\author[label4]{Reza Farzad} \ead{rfarzad@nd.edu}
\address[label4]{Department of Civil and Environmental Engineering \& Earth Sciences, University of Notre Dame, IN 46556, USA}
\address[label2]{Sonny Astani Department of Civil and Environmental Engineering, University of Southern California, Los Angeles, CA 90089, USA}

\author[label4]{Patrick T. Brewick} \ead{pbrewick@nd.edu}


\author[label2]{Erik A. Johnson}
\ead{JohnsonE@usc.edu}

\author[label3]{Steven F. Wojtkiewicz}
\address[label3]{Department of Civil and Environmental Engineering, Clarkson University, Potsdam, NY 13699, USA}
\ead{swojtkie@clarkson.edu}

\begin{abstract}

In computational mechanics, multiple models are often present to describe a physical system. 
While Bayesian model selection is a helpful tool to compare these models using measurement data, it requires the computationally expensive estimation of a multidimensional integral --- known as the marginal likelihood or as the model evidence (\ie the probability of observing the measured data given the model) --- over the multidimensional parameter domain.
This study presents efficient approaches 
for estimating this marginal likelihood by transforming it into a one-dimensional integral that is subsequently evaluated using a quadrature rule at multiple adaptively-chosen iso-likelihood contour levels.
Three different algorithms are proposed to estimate the probability mass at each adapted likelihood level using samples from importance sampling, stratified sampling, and Markov chain Monte Carlo sampling, respectively. The proposed approach is illustrated --- with comparisons to Monte Carlo, nested, and MultiNest sampling --- through four numerical examples. The first, an elementary example, shows the accuracies of the three proposed algorithms when the exact value of the marginal likelihood is known. 
The second example uses an 11-story building subjected to an earthquake excitation with an uncertain hysteretic base isolation layer with two models to describe the isolation layer behavior. 
The third example considers flow past a cylinder when the inlet velocity is uncertain. Based on the these examples, the method with stratified sampling is by far the most accurate and efficient method for complex model behavior in low dimension, particularly considering that this method can be implemented to exploit parallel computation. 
In the fourth example, the proposed approach is applied to heat conduction in an inhomogeneous plate with uncertain thermal conductivity modeled through a 100 degree-of-freedom Karhunen-Lo\`{e}ve expansion. 
The results indicate that 
MultiNest cannot efficiently handle the high-dimensional parameter space, whereas the proposed MCMC-based method more accurately and efficiently explores the parameter space.
The marginal likelihood results for the last three examples --- when compared with the results obtained from standard Monte Carlo sampling, nested sampling, and MultiNest algorithm --- show good agreement. 
\end{abstract}


\end{frontmatter}


\section{Introduction}
\label{sec:intro}
In computational mechanics, one often encounters the dilemma of choosing a model, for a physical system or phenomenon, that is computationally efficient as well as accurate. This model may be used for response prediction, design optimization, experiment design, and so forth. There are a variety of model selection criteria based on information theory, representing a trade-off between model's simplicity and the fit accuracy, including but not limited to Akaike information criterion (AIC) \cite{sakamoto1986akaike}, Bayesian-Schwarz information criterion (BIC) \cite{schwarz1978estimating}, the focused information criterion (FIC) \cite{Claeskens2003}, the minimum description length principle (MDL) \cite{Gruenwald2019}, and the deviance information criterion (DIC) \cite{Spiegelhalter2014}. DIC has been developed concerning the prediction accuracy of future datasets rather than considering model selection and, therefore, Bayesian evidence outperforms DIC in some cases. However, two common philosophies to guide the model choice across a set of candidate models are \emph{model falsification} \cite{popper2005logic,de2018investigation} and \emph{model selection} \cite{burnham2002model} (which has been shown to be equivalent to approximate Bayesian computation (ABC) \cite{Agnimitra2024}). In model falsification, models are eliminated using measurement data but no relative judgement is provided among the remaining unfalsified models. 
On the other hand, in \emph{Bayesian model selection}, one or more model(s) are selected from a larger candidate set \cite{de2018computationally,de2019hybrid,rao2001model,chipman2001practical} by ranking them based on posterior model probabilities; however, to estimate those probabilities, a multidimensional integral over the model parameter domain must be evaluated.
This model selection approach has been applied to a wide range of fields, \eg finance \cite{cremers2002stock}, genetics \cite{Crouse2022}, signal processing \cite{andrieu2001model}, subsurface flow modeling \cite{elsheikh2014efficient}, constitutive material modeling \cite{madireddy2015bayesian,asaadi2017computational}, 
structural dynamics \cite{beck2004model,beck2010bayesian,cheung2010calculation,muto2008bayesian,mthembu2011model}, and validation of structural models \cite{grigoriu2008solution}. 
In Bayesian model selection \cite{Jaynes2003}, 
 Bayes' theorem is applied to quantify the likelihood that each model could generate the measurement data. 
A common form of Bayesian model selection uses the \textit{Bayes factor}, which is a ratio of the marginal likelihoods of, or evidences for, two models \cite{kass1995bayes,goodman1999toward}. 
Note that the \textit{Occam's razor} principle, which suggests that models with lesser complexity should be favored among models of comparable accuracy, is 
also embedded in Bayesian model selection, as discussed in Beck and Yuen \cite{beck2004model}, MacKay \cite{mackay1992bayesian,mackay1992bayesiana}, and Gull \cite{gull1988bayesian}. A combined model selection and falsification approach can also be followed as shown in De \textit{et al.} \cite{de2019hybrid}; however, this approach still requires 
the estimation of marginal likelihood for a few models. 

The main computational effort in Bayesian model selection arises in the accurate estimation of the \textit{evidence}, the marginal likelihood for each model, requiring many forward model simulations. (Note that the terms ``evidence'' and ``marginal likelihood'' are used interchangeably herein.) A standard Monte Carlo estimation proves costly in this case when the number of model parameters increases.
A number of methods have been proposed in recent decades to intelligently select the random samples needed to estimate the marginal likelihood or evidence. For example, the posterior harmonic mean estimator \cite{newton1994approximate} samples from the posterior distributions of the parameters; this estimator, however, may have infinite variance \cite{raftery2006estimating}. 
Importance sampling \cite{kass1995bayes,Tran2021} and its adaptive variants \cite{Bugallo2017,Schuster2018} have also been applied to improve the efficiency of standard Monte Carlo sampling to estimate the marginal likelihood. 
Annealed importance sampling \cite{neal2001annealed,Li2023} and the power posterior method \cite{friel2008marginal} are closely related to thermodynamic integration and may also be used. Ching \textit{et al.} \cite{ching2007transitional} introduced the Transitional Markov chain Monte Carlo method in which samples, drawn from intermediate distributions, ultimately converge to a target distribution, providing an estimate of the marginal likelihood. Sandhu \textit{et al.} \cite{sandhu2014bayesian,sandhu2017bayesian} used posterior samples of model parameters obtained from Markov chain Monte Carlo using the Chib-Jeliazkov method \cite{chib2001marginal} for this estimation. Similarly, subset simulation methods \cite{botev2012efficient,ChiachioBeckChiachioRus2014,DiazDelaO_GarbunoInigo_Au_Yoshida_2017,VakilzadehHuangBeckAbrahamsson2017} and polynomial chaos approaches \cite{nagel2016spectral} may also be employed to estimate the marginal likelihood; however, the estimates provided by these methods can be biased as they are prone to using correlated random samples or approximated likelihood functions.   

Skilling \cite{skilling2006nested} proposed a method for the estimation of marginal likelihood that converts the multi-dimensional marginal likelihood integral into a one-dimensional integral representation.  This method, known as nested sampling, simplifies the integration task, 
but requires more samples from the high-likelihood region. 
MultiNest is a widely-used algorithm for efficient implementation of nested sampling for multi-modal posteriors \cite{Feroz_2009,Dittmann2024}. Considering the same mechanism used by MultiNest, Feroz \textit{et al.}~\citep{Feroz2013} presented the more efficient importance nested sampling method by using previously discarded samples. 
Polson and Scott \cite{polson2014vertical} outlined an approach called ``vertical-likelihood Monte Carlo'' based on a Lebesgue representation of the marginal likelihood that unifies nested sampling and other developments for marginal likelihood estimation; 
however, evaluating the resulting converted one-dimensional integral is not straightforward and incurs exorbitant computational costs for large-scale complex models with high-dimensional uncertainty. 
Therefore, with the advent of newer and more complex models to elaborate system behavior in greater depth, there is a need for novel and enhanced methods to efficiently explore stochastic parameter spaces in higher dimensions. 
Readers are referred to Llorente \textit{et al.}~\cite{Llorente2023} for an extensive and up-to-date review of marginal likelihood computation for model selection, including detailed descriptions, advantages, and weaknesses.

Herein, an adaptive approach is proposed to estimate the marginal likelihood. 
In this proposed approach, a probability integral transformation is first used to convert the multidimensional integration of marginal likelihood into a one-dimensional integration in an approach similar to that in ``vertical-likelihood Monte Carlo'' \cite{polson2014vertical}. 
The resulting one-dimensional integral is then evaluated using a quadrature rule; the quadrature points are calculated using an adaptive likelihood level approach, where samples with increasing levels of likelihoods are sequentially generated.
Three algorithms to efficiently generate these random samples are presented herein: based on importance sampling, stratified sampling, and Markov chain Monte Carlo sampling, respectively. 
In the first algorithm, samples for the current likelihood level are generated from an importance distribution formed using the samples from previous levels. In the second proposed algorithm, a subset of strata that contains samples from the previous level are used to generate samples for the current likelihood level. In the third algorithm, Markov chains are run starting from the previous level's samples to generate samples for the current likelihood level.
The proposed algorithms provide flexibility in choosing the likelihood levels to focus on regions with high likelihood values, thereby providing better computational efficiency relative to standard Monte Carlo sampling. Further, samples from the posterior model parameter distribution are not needed for the estimation; additionally, moments of the posterior distribution are generated as by product of these algorithms.

The proposed algorithms are illustrated using four numerical examples. First, an elementary example with a Gaussian likelihood and a Gaussian prior is used so that the true value of the evidence or marginal likelihood is known and accuracies can be assessed. 
The second example studies an 11-story base-isolated building with uncertainty in the nonlinear hysteretic isolation layer. Using roof acceleration measurements, the evidence or marginal likelihood is estimated for a linear approximation of the isolation layer and for two nonlinear model classes, namely, Bouc-Wen and bilinear hysteresis models. A comparison with Monte Carlo, nested sampling, and MultiNest is used to evaluate the results computed with the proposed approach. 
In the third example, the velocities of a two-dimensional flow are measured at several points downstream from a cylinder in a closed channel.
The inlet velocity profile is assumed to be parabolic with an uncertain maximum inlet velocity distributed as a truncated Gaussian. 
The marginal likelihood of the fluid flow model using the discretized Navier-Stokes equation and the uncertain inlet velocity is computed using the proposed algorithms and again compared with that from Monte Carlo and nested sampling. 
In the final example, an inhomogeneous plate with uncertain thermal conductivity is used to illustrate one of the proposed algorithms. The thermal conductivity is expressed using a random field represented by a Karhunen-Loeve expansion. This example has a parameter space that is of significantly higher dimension compared to the previous three examples. The results show closer agreement with Monte Carlo than those provided by nested sampling.

This paper is organized as follows. The next section provides brief backgrounds of Bayesian model selection and various sampling strategies. The three proposed algorithms are presented in Section~\ref{sec:method} with a discussion of issues related to their implementation in Section~\ref{sec:disc}. Then, in Section~\ref{sec:ex}, four numerical examples are used to illustrate the proposed algorithms, followed by conclusions in Section~\ref{sec:conc}. 

\section{Background}

\subsection{Definition of a Model}
A model to represent some physical phenomenon is defined here by a set of mathematical equations. Based on the characteristics of these equations, a model may be linear or nonlinear, and may describe static or dynamic behavior. The $k$\textsuperscript{th} model is parameterized by a set of parameters $\thetaa_k \in \mathbb{R}^{n_{\thetaa_k}}$ (the model number $k$ will be omitted subsequently for notational simplicity). 

\subsection{Bayesian Model Selection}
Let $\mathscr{M}=\{\M_1$, $\M_2$,$\dots$, $\M_{\Nm}\}$ be the set of $\Nm$ different models considered to describe a particular system.  Given a data set $\D$ containing measurements from the physical system, the goal of Bayesian model selection is to select the most plausible model(s) to represent the system. 
The \emph{posterior} model probabilities (\ie the probability of each model conditioned on measurement data $\D$) are given by Bayes' theorem:
\begin{equation} \label{eq:mod_sel_defn}
\begin{split}
\Prob{\Mk|\D}=\frac{\pdf(\D|\Mk)\Prob{ \Mk} }{ \pdf(\D)},\quad {k}=1, 2, \dots, \Nm.\\
\end{split}
\end{equation}
where $\Prob{ \Mk}$ is an \textit{a priori} measure of model plausibility assigned by the modeler based on past experience, normalized so that $\sum_{k=1}^{\Nm} \Prob{\Mk}=1$, and denominator
\begin{equation}\label{eq:denominator}
\pdf(\D)=\sum_{k=1}^{\Nm}\pdf(\D|\Mk)\Prob{\Mk}
\end{equation}
using the theorem of total probability. Herein, {$\Prob{\cdot}$} denotes a probability; {$\pdf(\cdot)$} denotes a probability density and the notation $\thetaa_j\sim\pdf(\thetaa)$ denotes choosing a sample $\thetaa_j$ according to the density $\pdf(\thetaa)$.

For a particular model $\Mk$, the model evidence or marginal likelihood $\Evd^{(k)}=\pdf( \D|\Mk)$ is
\begin{equation} \label{eq:evd}
\begin{split}
\Evd^{(k)}&=\int_{\Thetaa} \pdf(\D|\thetaa,\Mk)\pdf(\thetaa|\Mk)d\thetaa\\
&=\int_{\Thetaa} \Le(\thetaa,\Mk)\pdf(\thetaa|\Mk)d\thetaa
\end{split}
\end{equation}
where {$\thetaa\in\Thetaa \subseteq \mathbb{R}^{n_{\thetaa}}$} is the (uncertain) parameter vector, with prior probability $\pdf(\thetaa|\Mk)$ for model $\Mk$, and $\Le(\thetaa,\Mk)$ $=\pdf(\D|\thetaa,\Mk)$ is the likelihood function (\ie the data likelihood given a model and parameter vector).  
(It is assumed herein that the data are continuous quantities, but the approach can be adapted for discrete quantities by replacing densities $\pdf(\cdot)$ with probabilities $\Prob{\cdot}$; mixed continuous/discrete quantities can be accommodated as long as either the likelihood $\Le(\thetaa,\Mk)$ or the parameter prior $\pdf(\thetaa|\Mk)$ are finite for all $\thetaa\in\Thetaa$.)


\subsection{Evidence Computation via Nested Sampling}\label{sec:NSandMN}


As shown in \eqref{eq:evd}, computing the model evidence (marginal likelihood) $\Evd^{(k)}$ is more computationally demanding for models with a greater number of parameters, since it requires computing the high-dimensional integral over the entire parameter space domain. Nested sampling was introduced by Skilling \cite{skilling2006nested} to convert the evaluation of the evidence into a tractable one-dimensional integral.  While this has proven to be a powerful technique, the underlying method still relies upon Monte Carlo sampling, which is poorly suited to multi-modal parameter spaces.
To address this challenge, Shaw \textit{et al.} \cite{shaw2007efficient} proposed a clustered nested sampling method to form multiple ellipsoidal clusters to capture multi-modal posterior distributions. The MultiNest algorithm was subsequently developed \cite{feroz2008multimodal} and improved \cite{Feroz_2009} by Feroz \textit{et al.}~as a more advanced means of simultaneous ellipsoidal nested sampling. A brief review of the MultiNest algorithm is presented below; for complete details, the interested reader should consult Feroz \textit{et al.}~\cite{Feroz_2009}.

The MultiNest method is an iterative algorithm that is initialized by selecting $N$ ``active points'' by randomly sampling over the entire parameter space.  To account for the multi-modality, MultiNest assigns these samples to clusters (groups) 
bounded by iso-likelihood contours. 
A series of (possibly overlapping) ellipsoids are then simultaneously formed and bounded around each cluster of samples.  
The optimal ellipsoidal decomposition is performed using an ``expectation-maximization'' approach to minimize the total volume $\Vt$ of the ellipsoids while satisfying $\Vt > X_i/f$, where $X_i\approx\exp{(-i/N)}$ is the expected prior volume at the $i$th iteration of the algorithm
and $0<f<1$ is a user-defined value for the target efficiency. As $f$ increases, the minimum total volume decreases, and the EM algorithm chooses smaller $\Vt$, leading to faster algorithm convergence, but at the possible cost that the algorithm might not cover the full likelihood volume, resulting in biased estimates. Another issue is the possibility of an overshoot of the ellipsoidal decomposition in higher dimensions, leading to a remarkable decrease in sampling efficiency.

After constructing the ellipsoidal bounds, new samples are generated uniformly from within the bounded ellipsoids. More specifically, the probability of choosing ellipsoid $l$ among the $n_{\text{ell},i}$ ellipsoids found during the $i$th iteration
is $V_l/\Vt$, in which $\Vt= \sum_{l=1}^{n_{\text{ell},i}} V_l$. Then, if a new sample has a larger likelihood value than the minimum likelihood within the current samples, \ie $\mathcal{L_{\mathrm{new}}}>\mathcal{L}_{\mathrm{min}}$, the newly drawn sample is accepted to replace the minimum likelihood value with a probability equal to the inverse of the number of ellipsoids containing the point; this process leads to proper sampling from all ellipsoids by reducing the possibility of sampling too many points from highly-intersected regions.

%
Using the updated set of samples, another set of ellipsoids is formed and new samples are generated to replace the new minimum likelihood value(s).  This iterative process continues until some convergence criterion is satisfied, \eg a maximum number of iterations or a tolerance related to the evidence evaluation.  
The advantage of MultiNest is that it is flexible enough to explore a range of complicated posterior shapes through constructing non-overlapping and overlapping ellipsoids. 
In this study, the target efficiency parameter $f$, representing the fraction of the expected prior volume enclosed by the ellipsoidal bounds, is fixed at 10\%. 
The only hyper-parameter set for each application is the number of samples, which is specified for each application.



\subsection{Review of Sampling Methods}\label{sec:review}
The proposed algorithm will use three sampling methods, which are reviewed in this section assuming, for notational simplicity, a scalar parameter $\theta\in\Theta$. 
\subsubsection{Importance sampling}\label{sec:is}
Importance sampling is used to estimate an expectation $\mu_f=\Exp_{\pdf}[f(\theta)]$ when sampling from the density $\pdf(\theta)$ is difficult \cite{ross2013simulation}. Importance sampling instead draws $N$ samples $\{\theta_j\}_{j=1}^N$ from a similar density function $\pdfis(\theta)$, called the importance sampling density (ISD), and gives the unbiased estimate
\begin{equation}
\hat\mu_f^\mathrm{IS}=\frac{1}{N}\sum_{j=1}^N f(\theta_j)w_j
\end{equation}
where the \textit{importance weights} $w_j=\pdf(\theta_j)/\pdfis(\theta_j)$ are used to correct the bias introduced by sampling from $\pdfis(\theta)$. The variance of the estimator is given by 
\begin{equation}
\mathrm{Var}_{\pdfis}\left[\hat\mu_f^\mathrm{IS}\right]=\frac{1}{N}\left\{\Exp_{\pdf}\left[f^2(\theta)\frac{\pdf(\theta)}{\pdfis(\theta)}\right]-\mu_f^2\right\}
\end{equation}
The reduction in variance obtained compared to a standard Monte Carlo estimator $\hat\mu_f^{\text{MC}}=\frac{1}{N}\sum_{j=1}^N f(\theta_j)$ with $\theta_j\sim \pdf(\theta)$ is given by
\begin{equation}
\mathrm{Var}_{\pdf}\left[\hat\mu_f^\mathrm{MC}\right]-\mathrm{Var}_{\pdfis}\left[\hat\mu^\mathrm{IS}_f\right]=\frac{1}{N}\left\{\Exp_{\pdf}\Bigg[f(\theta)^2\left(1-\frac{\pdf(\theta)}{\pdfis(\theta)}\right)\Bigg]\right\}
\end{equation}
Hence, the use of importance sampling can produce variance reduction by choosing the importance density proportional to $|f(\theta)|\pdf(\theta)$. However, $\pdf(\theta)$ or $\pdfis(\theta)$ are often known up to a constant. In that case, one may use a normalized importance sampling (NIS), which estimates the expectation as
\begin{equation}
\hat\mu_f^\mathrm{NIS}=\sum_{j=1}^N f(\theta_j)\tilde w_j
\end{equation}
where the normalized weights $\tilde w_j$ are given by 
\begin{equation}
\tilde w_j=\frac{w_j}{\sum_{k=1}^N w_k}=\frac{\pdf(\theta_j)/\pdfis(\theta_j)}{\sum_{k=1}^N\pdf(\theta_k)/\pdfis(\theta_k)}
\end{equation}
This normalized importance sampling estimator is biased but consistent (\ie asymptotically unbiased).
Choosing $\pdfis(\theta)>0$ whenever $\pdf(\theta)>0$ also achieves variance reduction. A measure of effectiveness in using the importance density $\pdfis(\theta)$, the effective sample size (ESS), is given by
\begin{equation}
\textrm{ESS}=\frac{\left(\sum_{j=1}^N w_j\right)^2}{\sum_{j=1}^N w_j^2}
\end{equation}

\subsubsection{Stratified sampling}
Stratified sampling suggests dividing the sample space $\Theta$ into $\nst$ disjoint subspaces $\{\Theta^{(s)}\}_{s=1}^{\nst}$ where $\cup_{s=1}^{\nst}\Theta^{(s)}=\Theta$ and $\Theta^{(r)}\cap\Theta^{(s)}=\ouremptyset$ for $r\neq s$. The mean of the quantity $f$ is then estimated within each of these strata, denoted as $\hat\mu_f^{(s)}$ for $s=1,\dots,\nst$.
The strata means are then combined, using the theorem of total probability, to give
\begin{equation}
\Exp_{\pdf}[f]\approx\hat\mu_f^\mathrm{SS}=\sum_{s=1}^{\nst}\hat\mu_f^{(s)}p^{(s)}
\end{equation}
where $p^{(s)}=\Prob{\Theta^{(s)}}$.
The variance reduction compared to the standard Monte Carlo method is given by McKay \textit{et al.} \cite{mckay1979comparison}
\begin{equation}
\mathrm{Var}_{\pdf}\left[\hat\mu_f^\mathrm{MC}\right]-\mathrm{Var}_{\pdf|\Theta}\left[\hat\mu_f^\mathrm{SS}\right]=\frac{1}{N}\sum_{s=1}^{\nst}\left[\hat\mu_f^{(s)}-\hat\mu_f^\mathrm{SS}\right]^2p^{(s)}
\end{equation}
where each strata mean $\hat\mu_f^{(s)}$ is calculated using $N^{(s)} = p^{(s)}N$ samples for $s=1,\dots,\nst$.

\subsubsection{Markov chain {Monte Carlo (MCMC)}} 
Markov chain Monte Carlo is used to sample from a distribution otherwise difficult to sample \cite{chib1995understanding,chib2001markov}.
For this purpose, a Markov chain with 
stationary distribution $\pi(\cdot)$ 
is constructed to explore the sample space $\Theta$ and choose samples \cite{andrieu2003introduction}. 
The stationary distribution $\pi(\theta)$ will be $\pdf(\theta|\Le(\theta)>\lambda)$ for some likelihood level threshold $\lambda$ 
--- \ie proportional to $\pdf(\theta)$ but only in the region where the likelihood exceeds some level $\lambda$. 
The stationary distribution of the chain satisfies
\begin{equation}
\pi(\eta)\mathrm{d}\eta=\int_{\Theta} K(\theta,\mathrm{d}\eta)\pi(\theta)\mathrm{d}\theta
\end{equation}
where the transition kernel of the Markov chain is defined as \cite{tierney1994markov,chib1995understanding}
\begin{equation}\label{eq:kernel}
K(\theta,\mathrm{d}\eta)=f(\theta,\eta)\mathrm{d}\eta+\left[1-\int_{\Theta}f(\theta,\eta)\mathrm{d}\eta\right]\delta_{\theta}(\mathrm{d}\eta)
\end{equation}
for some transition function $f(\theta,\eta)$ with conditions that $f(\theta,\theta)=0$ and $\delta_\theta(\mathrm{d}\eta)=1$ for $\theta\in \mathrm{d}\eta$ or 0 for $\theta\notin \mathrm{d}\eta$. The probability of the chain staying at $\theta$ is $\left[1-\int_{\Theta}f(\theta,\eta)\mathrm{d}\eta\right]$. 
A popular algorithm for generating samples using MCMC is the \textit{Metropolis-Hastings (MH) algorithm}, which assumes that the transition from $\theta$ to $\eta$, for $\theta\neq \eta$, is of the form \cite{chib1995understanding}
\begin{equation}
f_{\mathrm{MH}}(\theta,\eta)=q(\theta,\eta)\alpha(\theta,\eta)
\end{equation}
where $q(\theta,\eta)$ is an assumed proposal density and $\alpha(\theta,\eta)$ is the acceptance rate defined by
\begin{equation}
\alpha(\theta,\eta)=\begin{cases}
\min\left[\frac{\pi(\eta)q(\eta,\theta)}{\pi(\theta)q(\theta,\eta)},1\right], \quad ~\!\text{for $\pi(\theta)q(\theta,\eta)>0$} \\
1, \qquad \qquad \qquad \quad\quad  \text{for $\pi(\theta)q(\theta,\eta)=0$}
\end{cases}
\end{equation}
However, with increasing dimension, the MH algorithm becomes inefficient. Au and Beck \cite{au2001estimation} proposed a modified algorithm, used in the numerical examples herein, with a higher acceptance of generated candidate samples by using the MH algorithm component-wise.


\section{Proposed Methodology}\label{sec:method}
\subsection{Probability integral transform}
Evidence integral \eqref{eq:evd}, omitting model variable {$\Mk$}, is rewritten in \eqref{eq:evidencenested} by: 
\begin{itemize}
	\item[(i)] expressing the likelihood
	$\Le(\thetaa)=\int_0^{_{\Le(\thetaa)}}\mathrm{d}\lambda$;
	\item[(ii)] rearranging the order of integration;
	\item[(iii)] defining a monotonically nonincreasing function $\chi(\lambda)$ that is the probability mass enclosed in the parameter space subset where likelihoods $\Le({\thetaa})$ exceed $\lambda$; 
	\item[(iv)] defining the monotonically nonincreasing inverse function $\varphi(\chi)$, the likelihood level that contains a given probability mass; \ie, $\varphi(\chi(\lambda))\equiv\lambda$;
	\item[(v)] changing the variable of integration; 
	\item[(vi)] approximating the integral by discretizing over $\chi$; and 
	\item[(vii)] defining $\varphi_i = \varphi(\chii) \equiv \lambda_i$.
\end{itemize}
%
\begin{equation}
\label{eq:evidencenested}
\begin{split}
\Evd&=\int_{\Thetaa} \pdf(\D|\thetaa)\pdf(\thetaa)d\thetaa\\
&=\int_{\Thetaa} \Le(\thetaa)\pdf(\thetaa)d\thetaa\\
&= \int_{\Thetaa} {\color{black}\underbrace{\color{black}\left[\int_0^{\Le(\thetaa)}\mathrm{d}\lambda\right]}_{\Le(\thetaa)}} \pdf(\thetaa)\mathrm{d}\thetaa\\
&= \int_{0}^{\infty} {\color{black}\underbrace{\color{black}\left[\int_{\Le(\thetaa)>\lambda}\pdf(\thetaa)\mathrm{d}\thetaa\right]}_{\chi(\lambda)}} \mathrm{d}\lambda\\
&= \int_0^\infty \chi (\lambda)\mathrm{d}\lambda
= \int_0^1 \varphi (\chi)\mathrm{d}\chi
\approx \sum_{i=1}^{i_\mathrm{final}} \varphi_i \Delta{\chii}
= \sum_{i=1}^{i_\mathrm{final}} \lambda_i \Delta{\chii}
\end{split}
\end{equation}
Once the multidimensional integral is converted into a one-dimensional integral, the problem becomes estimation of $\chii$ for a corresponding $\varphi_i\equiv \lambda_i$. 
\fref{fig:int} shows that the samples from successive likelihood levels are used to convert to $(\lambda,\chi)$ coordinates to estimate the integral using a quadrature rule.
Note that the authors previously showed \cite{de2018computationally} that the transformation \eqref{eq:evidencenested} is also the backbone of the nested sampling method \cite{skilling2006nested}. 


For efficient estimation of the $\chii$ to perform the quadrature, different variance reduction methods can be implemented; three such algorithms are proposed in the remainder of this section. Each algorithm is iterative with common stopping criteria that will be discussed in Section~\ref{sec:stop}. Additionally, the iterations begin with a low likelihood $\lambda_0$ close to zero and advance to greater likelihood levels $\lambda_i>\lambda_{i-1}$, as discussed in Section~\ref{sec:lik_level}. 

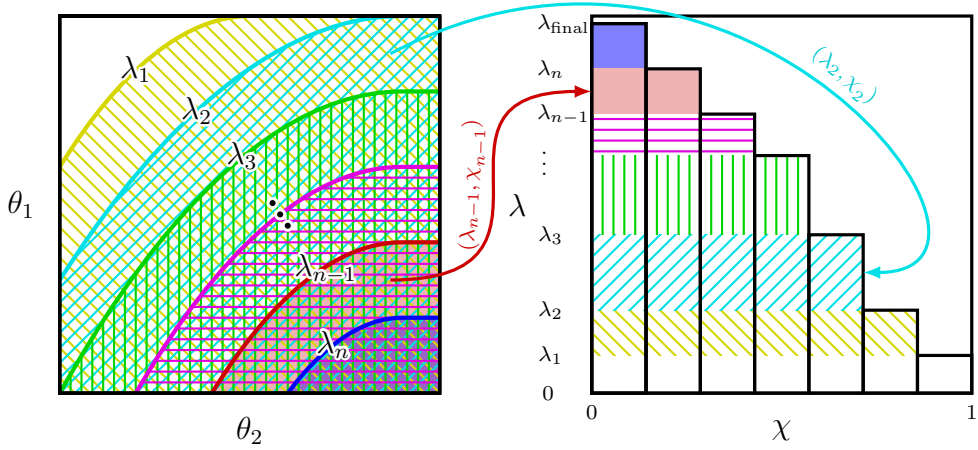
\begin{figure}
\centering
\begin{tikzpicture}[scale=1,every node/.style={minimum size=1cm},on grid]

\definecolor{myyellow}{rgb}{.84,.84,0} 
\definecolor{mycyan}{rgb}{0,.88,.9} 
\definecolor{mygreen}{rgb}{0,.83,0} 
\definecolor{mymagenta}{rgb}{.88,0,.88} 
\definecolor{myred}{rgb}{.8,0,0} 
\definecolor{myblue}{rgb}{0,0,1}
    \begin{scope}[
        yshift=0,every node/.append style={
        yslant=0.5,xslant=-1},yslant=0,xslant=0
                  ]

\path[save path=\contboundary] (0,0) rectangle (5,5);
\path [save path=\conta](3,0) parabola bend (4.5,1) (5,1)  ;
\path [save path=\contb](2,0) parabola bend (4.5,2) (5,2)  ;
\path [save path=\contc](1,0) parabola bend (4.5,3) (5,3)  ;
\path [save path=\contd](0,0) parabola bend (4.5,4) (5,4)  ;
\path [save path=\conte](0,1.5) parabola bend (4.5,5) (5,5)  ;
\path [save path=\contf](0,3) parabola bend (2.5,5) (2.5,5)  ;
\makeatletter\pgfsyssoftpath@setcurrentpath{\conta}\makeatother \path [save path=\contaa] -- (5,0) -- cycle;
\makeatletter\pgfsyssoftpath@setcurrentpath{\contb}\makeatother \path [save path=\contbb] -- (5,0) -- cycle;
\makeatletter\pgfsyssoftpath@setcurrentpath{\contc}\makeatother \path [save path=\contcc] -- (5,0) -- cycle;
\makeatletter\pgfsyssoftpath@setcurrentpath{\contd}\makeatother \path [save path=\contdd] -- (5,0) -- cycle;
\makeatletter\pgfsyssoftpath@setcurrentpath{\conte}\makeatother \path [save path=\contee] -- (5,0) -- (0,0) -- cycle;
\makeatletter\pgfsyssoftpath@setcurrentpath{\contf}\makeatother \path [save path=\contff] -- (5,5) -- (5,0) -- (0,0) -- cycle;

\makeatletter\pgfsyssoftpath@setcurrentpath{\contbb}\makeatother \fill[myred!30];
\makeatletter\pgfsyssoftpath@setcurrentpath{\contaa}\makeatother \fill[myblue!60!myred!60!white];

\begin{scope}
\makeatletter\pgfsyssoftpath@setcurrentpath{\contboundary}\makeatother \clip[];
\end{scope}

\begin{scope}
\makeatletter\pgfsyssoftpath@setcurrentpath{\contff}\makeatother \clip[];
\foreach \x in {.2,.4,...,10} {\draw[myyellow,xslant=-1,thick] (\x,0) -- (\x,5);}
\end{scope}

\begin{scope}
\makeatletter\pgfsyssoftpath@setcurrentpath{\contee}\makeatother \clip[];
\foreach \x in {-4.8,-4.6,...,5} {\draw[mycyan,xslant=1,thick] (\x,5) -- (\x,0);}
\end{scope}

\begin{scope}
\makeatletter\pgfsyssoftpath@setcurrentpath{\contdd}\makeatother \clip[];
\foreach \x in {0.142857142857143,0.285714285714286,...,5} {\draw[mygreen,thick] (\x,0) -- (\x,5);}
\end{scope}

\begin{scope}
\makeatletter\pgfsyssoftpath@setcurrentpath{\contcc}\makeatother \clip[];
\foreach \x in {0.142857142857143,0.285714285714286,...,5} {\draw[mymagenta,thick] (0,\x) -- (5,\x);}
\end{scope}

\begin{scope}
\makeatletter\pgfsyssoftpath@setcurrentpath{\contboundary}\makeatother \clip[]; 
\makeatletter\pgfsyssoftpath@setcurrentpath{\conta}\makeatother \draw [ultra thick,myblue];
\makeatletter\pgfsyssoftpath@setcurrentpath{\contb}\makeatother \draw [ultra thick,myred];
\makeatletter\pgfsyssoftpath@setcurrentpath{\contc}\makeatother \draw [ultra thick,mymagenta];
\makeatletter\pgfsyssoftpath@setcurrentpath{\contd}\makeatother \draw [ultra thick,mygreen];
\makeatletter\pgfsyssoftpath@setcurrentpath{\conte}\makeatother \draw [ultra thick,mycyan];
\makeatletter\pgfsyssoftpath@setcurrentpath{\contf}\makeatother \draw [ultra thick,myyellow];
\end{scope}

\makeatletter\pgfsyssoftpath@setcurrentpath{\contboundary}\makeatother \draw[black,very thick] (0,0) rectangle (5,5);

    \end{scope} 

\node at (1,4.3) {\contour{white}{$\lambda_1$}};    
\node at (1.8,3.75) {\contour{white}{$\lambda_2$}};    
\node at (2.4,3.15) {\contour{white}{$\lambda_3$}};    
\foreach \x in {-.1,0,.1} {
\fill[white] (2.9+\x,2.37-1.5*\x) circle (0.08); 
\fill[black] (2.9+\x,2.37-1.5*\x) circle (0.04); 
}

\node at (3.6,0.65) {\contour{white}{$\lambda_n$}};    
\node at (3.5,1.65) {\contour{white}{$\lambda_{n-1}$}};    

 \fill[black]
      (0,2.5) node [left] {$\theta_1$}
         (2.5,-1) node [above] {$\theta_2$};	
         
         \draw[black,very thick] (7,0) rectangle (12,5);

\pgfmathsetmacro{\psimax}{4.90}
\pgfmathsetmacro{\psia}{4.30}
\pgfmathsetmacro{\psib}{3.70}
\pgfmathsetmacro{\psic}{3.15}
\pgfmathsetmacro{\psid}{2.10}
\pgfmathsetmacro{\psie}{1.10}
\pgfmathsetmacro{\psif}{0.50}

\fill[myblue!50] (7,\psia) rectangle (7+1*5/7,\psimax); 
\fill[myred!30] (7,\psib) rectangle (7+2*5/7,\psia);

\begin{scope}
\clip[] (7,0) rectangle (7+7*5/7,\psif);
\end{scope}

\begin{scope}
\clip[] (7,\psif) rectangle (7+6*5/7,\psie); 
\foreach \x in {.2,.4,...,10} {\draw[myyellow,xslant=-1,thick] (7+\x,0) -- (7+\x,5);}
\end{scope}

\begin{scope}
\clip[] (7,\psie) rectangle (7+5*5/7,\psid); 
\foreach \x in {-4.8,-4.6,...,5} {\draw[mycyan,xslant=1,thick] (7+\x,0) -- (7+\x,5);}
\end{scope}

\begin{scope}
\clip[] (7,\psid) rectangle (7+4*5/7,\psic); 
\foreach \x in {0.142857142857143,0.285714285714286,...,5} {\draw[mygreen,thick] (7+\x,0) -- (7+\x,5);}
\end{scope}

\begin{scope}
\clip[] (7,\psic) rectangle (7+3*5/7,\psib); 
\foreach \x in {0.062857142857143,0.205714285714286,...,5} {\draw[mymagenta,thick] (7+0,\x) -- (7+5,\x);}
\end{scope}

\begin{scope}
\clip[] (7,\psib) rectangle (7+2*5/7,\psia); 
\end{scope}

\pgfmathsetmacro{\ytickhorizloc}{6.12}
\node[font=\scriptsize,anchor=west] at (\ytickhorizloc,0) {$0$~~~~~};
\node[font=\scriptsize,anchor=west] at (\ytickhorizloc,\psif) {$\lambda_{1\hphantom{-1}}$};
\node[font=\scriptsize,anchor=west] at (\ytickhorizloc,\psie) {$\lambda_{2\hphantom{-1}}$};
\node[font=\scriptsize,anchor=west] at (\ytickhorizloc,\psid) {$\lambda_{3\hphantom{-1}}$};
\node[font=\scriptsize,anchor=west] at (\ytickhorizloc,\psic) {$\vdots$~~~~~~};
\node[font=\scriptsize,anchor=west] at (\ytickhorizloc,\psib) {$\lambda_{n-1}$};
\node[font=\scriptsize,anchor=west] at (\ytickhorizloc,\psia) {$\lambda_{n\hphantom{-1}}$};
\node[font=\scriptsize,anchor=west] at (\ytickhorizloc,\psimax) {$\lambda_\mathrm{final}$};

\draw[black,very thick] (7+0*5/7,0) rectangle (7+1*5/7,\psimax);
\draw[black,very thick] (7+1*5/7,0) rectangle (7+2*5/7,\psia);
\draw[black,very thick] (7+2*5/7,0) rectangle (7+3*5/7,\psib);
\draw[black,very thick] (7+3*5/7,0) rectangle (7+4*5/7,\psic);
\draw[black,very thick] (7+4*5/7,0) rectangle (7+5*5/7,\psid);
\draw[black,very thick] (7+5*5/7,0) rectangle (7+6*5/7,\psie);
\draw[black,very thick] (7+6*5/7,0) rectangle (7+7*5/7,\psif);
\draw[black,very thick] (7,0) rectangle (7+1*5/7,\psimax);

 \fill[black]
      (\ytickhorizloc-.1,2.5) node [] {$\lambda$}
         (9.5,-1) node [above] {$\chi$};	
          \fill[black]
      (7,-.25) node [font=\scriptsize] {0}
         (12,-.25) node [font=\scriptsize] {1};	

\draw[-latex,very thick,myred] (4.35,1.5) .. controls (7.0,1.5) and (4.5,\psia/2+\psib/2) .. (7,\psia/2+\psib/2) node [midway, above,sloped,inner sep=1mm,minimum height=0mm,font=\scriptsize] {$(\lambda_{n-1},\chisub{n-1})$};
\draw[-latex,very thick,mycyan] (4.35,4.5) .. controls (9,7) and (7+5*5/7+2.5,\psid/2+\psie/2) .. (7+5*5/7,\psid/2+\psie/2) node [midway, above,sloped,inner sep=1mm,minimum height=0mm,font=\scriptsize,color=mycyan!90!black] {$(\lambda_2,\chisub{2})$};



\end{tikzpicture}
\caption{Likelihood level adapted estimation of marginal likelihood. Note that, $\lambda_1<\lambda_2<\dots<\lambda_{n-1}<\lambda_n$.}\label{fig:int}
\end{figure}

\subsection{Likelihood level adapted importance sampling (LLA-IS)}
\label{sec:mlis}

The first method employs importance sampling to estimate $\chii$, which can be written
\begin{equation}\label{eq:chi_nis}
\chii=\Prob{\thetaa\in\widetilde{\Thetaa}_i}=\Exp\left[\mathbb{I}_{\widetilde{\Thetaa}_i}(\thetaa)\right]
\end{equation}
where the parameter domain $\widetilde\Thetaa_i$ is the region in which the likelihood is above $\lambda_i$, \ie
 \begin{equation}\label{eq:Theta_def}
 \widetilde\Thetaa_i=\{\thetaa|\Le(\thetaa)>\lambda_i,\thetaa\in\Thetaa\}
 \end{equation}
 for positive integer $i\in\mathbb{Z}^+$, and where the indicator function $\mathbb{I}_{(\cdot)}(\cdot)$ is defined by
\begin{equation}
\mathbb{I}_{\widetilde{\Thetaa}_i}(\thetaa)=\begin{cases}
1 \qquad \text{if $\thetaa \in \widetilde{\Thetaa}_i$}\\0 \qquad \text{otherwise }
\end{cases}.
\end{equation}
Hence,
the quantity $\chii$ in \eqref{eq:chi_nis} can be approximated with the estimator
\begin{equation}\label{eq:lla_is_chi}
\widehatchisub{i}^{\mathrm{IS}}= \frac{1}{N}\sum_{j=1}^N\mathbb{I}_{\widetilde{\Thetaa}_i}(\thetaa_{j,i})w_{j,i}
\end{equation}
where $w_{j,i}=\pdf(\thetaa_{j,i})/\pdfis(\thetaa_{j,i})$; $\pdfis(\cdot)$ is the importance sampling density with $\pdfis(\thetaa)>0$ whenever $\pdf(\thetaa)>0$. 

The main challenge applying this algorithm lies in choosing the appropriate form for the importance density $\pdfis(\thetaa)$.
It may be chosen as a normal distribution with its mean at the posterior mode $\hat\thetaa$ guessed from the previous set of samples and arbitrarily large variance $\widehat\Sigm$ \cite{mukherjee2006nested}.
However, the spread of the importance density should be chosen carefully as a small spread can result in an erroneous estimation whereas a large spread will result in a slow convergence or non-convergence. 
Herein, the mean of this Gaussian importance density is fixed at the mean of the retained samples and, unless otherwise specified, the standard deviation at twice the standard deviation estimated using the retained samples. 
A challenge for LLA-IS is that finding a good importance sampling density in high-dimensional spaces can become increasingly difficult as well as computationally expensive, so it may suffer from the curse of dimensionality.

An algorithm for likelihood level adapted estimation of marginal likelihood (evidence) using importance sampling is presented in Algorithm \ref{alg:mlnis}. The algorithm starts with $\widehatchisub{0}^{\mathrm{IS}}=1$ and a small $\lambda_0$ that is very small or equal to zero. The importance density at the start is same as the parameter prior. The algorithm chooses likelihood threshold $\lambda_i$, greater than $\lambda_{i-1}$ in the previous iteration, after the simulation of $N$ samples for practicality, as discussed subsequently in Section~\ref{sec:lik_level}.  Then, the corresponding $\widehatchisub{i}^{\mathrm{IS}}$ is estimated using the importance weight as shown in \eqref{eq:lla_is_chi}. 
At the end of every iteration, the evidence is updated using a quadrature rule (the algorithms herein use a rectangular quadrature, though they could be easily adapted toward higher-order schemes). 
Additional samples are drawn from the current importance density if the effective sample size falls below a pre-chosen threshold $\gamma_\mathrm{S}$, which is set to $0.5N$ herein (though this never occurred for the examples explored herein in Section~\ref{sec:ex}); a large threshold $\gamma_\mathrm{S}$ increases the computational budget as new random samples are more often added to the sample pool but a small $\gamma_\mathrm{S}$, on the other hand, will lead to slow convergence. 
In \fref{fig:mlis}, the increasing likelihood contours as well as a schematic of importance densities formed near the high-likelihood region are shown as the algorithm progresses.  

\begin{figure}
	\centering
	\begin{tikzpicture}[scale=0.75,every node/.style={minimum size=1cm},on grid,/tikz/fill between/on layer=main]
	
	\begin{scope}[
	yshift=-170,every node/.append style={
		yslant=0.5,xslant=-1},yslant=0.5,xslant=-1
	]

	
	\draw [ultra thick,red,name path=cont1](3,0) parabola bend (4.5,1) (5,1)  ;
	\draw [ultra thick,red!70,name path=cont2](2,0) parabola bend (4.5,2) (5,2)  ;
	\draw [ultra thick,red!40,name path=cont3](1,0) parabola bend (4.5,3) (5,3)  ;
	\draw [ultra thick,yellow,name path=cont4](0,0) parabola bend (4.5,4) (5,4)  ;
	\draw [ultra thick,green!60,name path=cont5](0,1) parabola bend (4.5,5) (5,5)  ;
	\draw [ultra thick,teal,name path=cont6](0,2) parabola bend (3,5) (3,5)  ;
	\draw [ultra thick,cyan,name path=cont7](0,3) parabola bend (2,5) (2,5)  ;
	\draw [ultra thick,blue,name path=cont8](0,4) parabola bend (1,5) (1,5)  ;

	\draw [thick,name path=cont0](3,0) -- (5,0) -- (5,1);
	\draw [thick,name path=cont9](0,4) -- (0,5) -- (1,5);

	\tikzfillbetween[
	of=cont1 and cont2
	] {pattern=grid, pattern color=red!60};
	\tikzfillbetween[
	of=cont2 and cont3
	] {pattern=grid, pattern color=orange};
	\tikzfillbetween[
	of=cont0 and cont1
	] {pattern=grid, pattern color=red};
	\tikzfillbetween[
	of=cont3 and cont4
	] {pattern=grid, pattern color=yellow!70!orange};
	\tikzfillbetween[
	of=cont4 and cont5
	] {pattern=grid, pattern color=green!60!yellow};
	\tikzfillbetween[
	of=cont5 and cont6
	] {pattern=grid, pattern color=blue!20!green};
	\tikzfillbetween[
	of=cont6 and cont7
	] {pattern=grid, pattern color=blue!50!green};
	\tikzfillbetween[
	of=cont7 and cont8
	] {pattern=grid, pattern color=green!10!blue};
	\tikzfillbetween[
	of=cont8 and cont9
	] {pattern=grid, pattern color=black!10!blue};

	\draw[black,very thick] (0,0) rectangle (5,5);
	
	\end{scope} 
	\begin{scope}[
	yshift=-83,every node/.append style={
		yslant=0.5,xslant=-1},yslant=0.5,xslant=-1
	]

	\fill[white,fill opacity=0.75] (0,0) rectangle (5,5);
	
	\draw [thick,name path=cont1,dashed](3,0) parabola bend (4.5,1) (5,1)  ;
	\draw [thick,name path=cont2,dashed](2,0) parabola bend (4.5,2) (5,2)  ;
	\draw [thick,name path=cont3,dashed](1,0) parabola bend (4.5,3) (5,3)  ;
	\draw [thick,name path=cont4,dashed](0,0) parabola bend (4.5,4) (5,4)  ;
	\draw [thick,name path=cont5,dashed](0,1) parabola bend (4.5,5) (5,5)  ;
	\draw [thick,name path=cont6,dashed](0,2) parabola bend (3,5) (3,5)  ;
	\draw [thick,name path=cont7,dashed](0,3) parabola bend (2,5) (2,5)  ;
	\draw [thick,name path=cont8,dashed](0,4) parabola bend (1,5) (1,5)  ;
	\draw[black,very thick] (0,0) rectangle (5,5);
	%

	\draw[red,fill=red!50,opacity=0.2] (2.2,2.2) circle (2cm);
	\draw[red,fill=red!50,opacity=0.9] (2.2,2.2) circle (0.05cm);
		\draw[red,fill=red!50,opacity=0.2] (2.5,1.5) circle (1.25cm);
		\draw[red,fill=red!50,opacity=0.9] (2.5,1.5) circle (0.05cm);
				\draw[red,fill=red!50,opacity=0.2] (2.85,1) circle (0.75cm);
				\draw[red,fill=red!50,opacity=0.9] (2.85,1) circle (0.05cm);

	\end{scope}

	\draw[-latex,thick](4,-5)node[right]{Iso-likelihood contours}
	to[out=180,in=90] (1.95,-4.75);

    \contourlength{1.5pt}%
	\node at (-4.15,-3.5) {\contour{white}{$\lambda_1$}};    
	\node at (-3.40,-3.5) {\contour{white}{$\lambda_2$}};    
	\node at (-2.65,-3.5) {\contour{white}{$\lambda_3$}};    
	
	\node at (3.1,-3.5) {\contour{white}{$\lambda_n$}};    
\node at (2,-3.5) {\contour{white}{$\cdots$}};    

\node at (0.6,-3.5) {\contour{white}{$\lambda_i$}};   
\node at (-0.9,-3.5) {\contour{white}{$\cdots$}};    
	
	\fill[black]
	(-2.5,-5.8) node [above] {$\theta_1$}
	(2.4,-5.9) node [above] {$\theta_2$};

	\end{tikzpicture}
	\caption{Likelihood level adapted importance sampling (LLA-IS) method: the iso-likelihood contours are shown with $\lambda_1<\lambda_2<\dots<\lambda_n$; importance densities are formed successively to generate samples from high likelihood region.} \label{fig:mlis}
\end{figure}
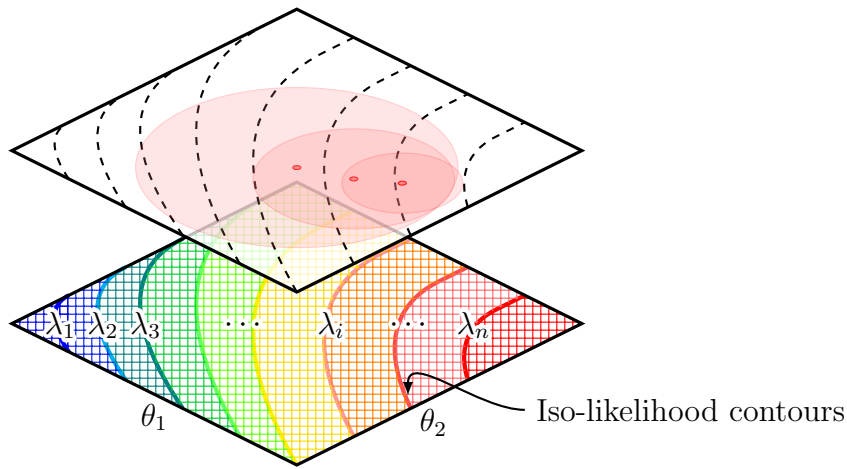

\begin{algorithm}
	\DontPrintSemicolon
	\CommentSty{\color{newblue}}
	{\bf{Initialization:}} Set $\widehatchisub{0}^{\mathrm{IS}}=1$, $\widehat\Evd_0^{\mathrm{~\!IS}}=0$, $\pdfis_0(\cdot)=\pdf(\cdot)$, and $\lambda_0 = 0$; choose a suitable $\gamma_\mathrm{S}$ based on a computational budget (\eg $\gamma_\mathrm{S}=0.5N$)\; 
	Set $i=1$\;
	\While{Stopping Criterion = FALSE}{
		Draw samples $\thetaa_{j,i}\sim\pdfis_{i-1}(\thetaa)$\textsuperscript{$\dagger$} for $j=1,\dots,N$\;
		Select a suitable new likelihood level $\lambda_i>\lambda_{i-1}$ (see Section~\ref{sec:lik_level})\;
		Define the likelihood exceedance region $\widetilde\Thetaa_i=\{\thetaa|\Le(\thetaa)>\lambda_i,\thetaa\in\Thetaa\}$\;
		Calculate likelihood values $\Le(\thetaa_{j,i})$, $j=1,\dots,N$, for these samples\;
		Evaluate the importance weights $w_{j,i}=\pdf(\thetaa_{j,i})/\pdfis_i(\thetaa_{j,i})$, $j=1,\dots,N$\;
		\uIf{$\mathrm{ESS}=\left(\sum_{j=1}^N w_{j,i}\right)^2 / \sum_{j=1}^N w_{j,i}^2 >\gamma_\mathrm{S}$}{
			$\widehatchisub{i}^{\mathrm{IS}}\approx \frac{1}{N}\sum_{j=1}^N\mathbb{I}_{\widetilde{\Thetaa}_i}(\thetaa_{j,i})w_{j,i}$\;}
		\Else
		{Draw another $\Delta N = \lceil \gamma_\mathrm{S} - \mathrm{ESS}  \rceil$ samples from density $\pdfis_i(\thetaa)$\;
			$N\leftarrow N + \Delta N$\;
            For each of the $\Delta N$ new samples, compute likelihoods as on line {7} and weights as on line {8}\; 
            Go to line {9}\;}
		Update the marginal likelihood $\widehat\Evd_i^{\mathrm{~\!IS}}=\widehat\Evd_{i-1}^{\mathrm{~\!IS}}+\lambda_i\left(\widehatchisub{i-1}^{\mathrm{IS}}-\widehatchisub{i}^{\mathrm{IS}}\right)$\;
		Assume\textsuperscript{$\dagger$} a proper importance density $\pdfis_i(\cdot)$ based on the samples $\thetaa_{1,i},\dots,\thetaa_{N,i}$ from this iteration\;
		$i\leftarrow i+1$\;
	}
	\KwResult{The marginal likelihood $\widehat\Evd^\mathrm{~\!IS}\subscript{final}$}
	%
	%
	\vspace{2pt}
	\caption{Likelihood level adapted marginal likelihood estimation using importance sampling (LLA-IS).}\label{alg:mlnis}
    \algorithmfootnote{%
	\textsuperscript{$\dagger$}\:Herein, for $i \ge 1$, a Gaussian distribution is assumed for $\pdfis_i(\cdot)$, with a mean the same as the mean of the samples with $\Le(\theta_{j,i})>\lambda_i$ and a 
    standard deviation that is twice the sample standard deviation, as discussed in Section~\ref{sec:mlis}.\\
	}
\end{algorithm}

\clearpage
\subsection{Likelihood level adapted stratified sampling (LLA-SS)}

The second algorithm proposed here implements stratified sampling in the likelihood level adapted approach,
estimating $\Prob{\thetaa\in\widetilde{\Thetaa}_i}$ by focusing on the strata where at least one sample $\thetaa_{j,i}$ has $\Le(\thetaa_{j,i})>\lambda_i$. \fref{fig:mlss} shows a hypothetical set of strata, highlighting those containing at least one sample with likelihood greater than some level $\lambda_i$.

Let the sample space $\Thetaa$ be divided into $\nst$ disjoint strata $\Thetaa^{(1)}$, $\Thetaa^{(2)}$, $\dots$, $\Thetaa^{(\nst)}$ such that $\cup_{s=1}^{\nst} \Thetaa^{(s)}=\Thetaa$ and $\Thetaa^{(r)}\cap\Thetaa^{(s)}=\ouremptyset$ for $r\neq s$. In iteration $i$, which will target samples whose likelihoods exceed a level $\lambda_i$, samples will only be drawn from the subset of strata, denoted by the strata index set $\mathcal{I}_{i-1}$, that contained one or more samples with likelihoods exceeding the likelihood level $\lambda_{i-1}$ during the previous iteration $(i-1)$; for the first iteration, $\mathcal{I}_0$ contains the indices of all strata. A new likelihood level $\lambda_i$ is chosen larger than the previous $\lambda_{i-1}$ (see Section~\ref{sec:lik_level}). 

In iteration $i$, $N_i^{(s)}$ additional samples are added in stratum $s$, resulting in the cumulative number of samples in stratum $s$ through iteration $i$ to be $N_{i,\mathrm{cum}}^{(s)}=\sum_{k=1}^i N_k^{(s)}$; \ie all of the previous samples from stratum $s$ are retained and re-used.
Thus, in each stratum $s$ during iteration $i$, draw $N_i^{(s)}$ additional samples 
$\thetaa_{j,i}^{(s)}\sim \pdf(\thetaa)$ but restricted only to those in $\Thetaa^{(s)}$ (see below for one approach to this selective sampling), and compute their corresponding likelihoods $\Le(\thetaa_{j,i}^{(s)})$, for $j=N_{i-1,\mathrm{cum}}^{(s)}+1,\dots,N_{i,\mathrm{cum}}^{(s)}$, where $N_{0,\mathrm{cum}}^{(s)} \equiv 0$. Then, the fraction of probability mass in stratum $s$ that exceeds likelihood $\lambda_i$ is estimated with 
\begin{equation}\label{eq:SS_chi_s_estimate}
\widehatchisub{i}^{(s)}=\frac{1}{N_{i,\mathrm{cum}}^{(s)}}\sum_{j=1}^{N_{i,\mathrm{cum}}^{(s)}}\mathbb{I}_{\widetilde{\Thetaa}_i}\left(\thetaa_{j,i}^{(s)}\right)
\end{equation}
If $p^{(s)}=\Prob{\thetaa\in \Thetaa^{(s)}}$ is the total probability mass in stratum $s$, then the combined estimate of the probability mass with likelihood exceeding $\lambda_i$ is 
\begin{equation}
\widehatchisub{i}^{\mathrm{SS}} = \sum_{s\in \mathcal{I}_{i-1}} p^{(s)} \widehatchisub{i}^{(s)}
\end{equation}
The mean of this estimate is $\Exp\left[\widehatchisub{i}^{\mathrm{SS}}\right]=\chii$ and its variance is
\begin{equation}
\mathrm{Var}\left[\widehatchisub{i}^{\mathrm{SS}}\right]=\frac{\sigma^2_{\chii}}{N_{i,\mathrm{cum}}}-\frac{1}{N_{i,\mathrm{cum}}}\sum_{s\in\mathcal{I}_{i-1}} \left[\widehatchisub{i}^{(s)}-\widehatchisub{i}^{\mathrm{SS}}\right]^2p^{(s)}
\end{equation}
where $N_{i,\mathrm{cum}}=\sum_{s=1}^{\nst}N_{i,\mathrm{cum}}^{(s)}$ and $\sigma^2_{\chii}$ is the variance of $\chi(\lambda_i)$.
\begin{algorithm}
	\DontPrintSemicolon
	{\bf{Initialization:}} Set $\widehatchisub{0}^\mathrm{SS}=1$, $\widehat\Evd_0^\mathrm{SS}=0$, and $\lambda_0=0$\;
	Divide the sample space into $\nst$ strata $\left\{\Thetaa^{(s)}\right\}_{s=1}^{\nst}$\;
	Set the initial strata index set $\mathcal{I}_{0}=\{1,\dots,\nst\}$ to contain all strata.\;
	Set $i=1$\;
	\While{Stopping Criterion = FALSE}{
		\For{$s\in\mathcal{I}_{i-1}$}{
			Set $N_i^{(s)}=N/\lvert \mathcal{I}_{i-1} \rvert$, where $|\cdot|$ here denotes the number of strata in the set\;
            Set $N_{i,\mathrm{cum}}^{(s)}=\sum_{k=1}^i N_k^{(s)}$\;
			Draw $N_i^{(s)}$ additional  
            samples $\thetaa_{j,i}^{(s)}\sim\pdf(\thetaa)$ from stratum $\Thetaa^{(s)}$  for $j=N_{i-1,\mathrm{cum}}^{(s)}+1,\dots,N_{i,\mathrm{cum}}^{(s)}$ where $N_{0,\mathrm{cum}}^{(s)}\equiv 0$\;
			Calculate likelihood values $\Le\left(\thetaa_{j,i}^{(s)}\right)$ for these samples\;
			Assign weights $w_i^{(s)}=p^{(s)}/N_{i,\mathrm{cum}}^{(s)}$ where $p^{(s)}=\Prob{\thetaa\in\Thetaa^{(s)}}$\;
		}
		Select a suitable new likelihood level $\lambda_i>\lambda_{i-1}$ (see Section~\ref{sec:lik_level})\; 
		Estimate $\widehatchisub{i}^{\mathrm{SS}}= \sum_{s\in\mathcal{I}_{i-1}}w_i^{(s)}\sum_{j=1}^{N_{i,\mathrm{cum}}^{(s)}}\mathbb{I}_{\widetilde{\Thetaa}_i}\left(\thetaa_{j,i}^{(s)}\right)$\; 
		Update the marginal likelihood $\widehat\Evd_i^\mathrm{SS}=\widehat\Evd_{i-1}^\mathrm{SS}+\lambda_i\left(\widehatchisub{i-1}^{\mathrm{SS}}-\widehatchisub{i}^{\mathrm{SS}}\right)$\;
		Let $\mathcal{I}_{i} = \left\{ s \big\lvert \max_j \Le\left(\thetaa_{j,i}^{(s)}\right)>\lambda_i \right\}$, 
        which will be used in the next iteration, be the set of all strata with at least one sample $\thetaa_{j^*,i}^{(s)}$ having likelihood greater than $\lambda_i$\; 
		$i\leftarrow i+1$\;
	}
	\KwResult{The marginal likelihood $\widehat\Evd\subscript{final}^\mathrm{SS}$}
	%
	%
	\vspace{2pt}
	\caption{Likelihood level adapted marginal likelihood estimation using stratified sampling (LLA-SS).}\label{alg:mlss}
\end{algorithm}
Finally, the marginal likelihood \eqref{eq:evd} is estimated, using a rectangle-rule integration, with
\begin{equation}
\widehat{\Evd}^{\mathrm{SS}} = \sum_i\lambda_i \left(\widehatchisub{i-1}^{\mathrm{SS}}-\widehatchisub{i}^{\mathrm{SS}}\right)
\end{equation}

Algorithm \ref{alg:mlss} shows the steps for this stratified sampling likelihood-level adapted estimation of the marginal likelihood $\Evd$, where the initial strata index set $\mathcal{I}_0$ includes all strata but, as the iterations progress, samples are obtained from a smaller number of strata that contain likelihood values larger than prior likelihood levels. \fref{fig:mlss} depicts a schematic of the LLA-SS algorithm for a two-dimensional parameter vector $\thetaa$, assuming a grid of $\nstk{k}$ substrata in dimension $k$ so that $\nst=\prod_{k=1}^{\nth}\nstk{k}$, showing a shaded region that contains the strata with $\Le(\thetaa)>\lambda_i$ from which additional samples are drawn. 
Implementing stratified sampling can be difficult with high-dimensional parameter spaces (\ie large $\nth$). However, as this LLA-SS iterates, only a very few strata will remain with $\Le(\thetaa)>\lambda_i$.


One way to construct the strata when the domain is gridded in each dimension and the elements of $\thetaa$ are independent
is to sample $\thetaa = [\theta_1, \theta_2, \dots, \theta_{\nth}]^\mathrm{T}$ from the joint probability density $\pdf(\thetaa) = \prod_{k=1}^{\nth}\pdf(\theta_k)$ using the inverses of the marginal cumulative density functions $F_{\theta_k}(\cdot)$: for stratum $s$, composed of substratum $s_k \in \{1,\dots,\nstk{k}\}$ in dimension $k$ for $k=1,2,\dots,\nth$, let 
\begin{equation}
\Thetaa^{(s)} = \Theta^{(s_1)}_1\times \Theta^{(s_2)}_2\times\dots\times \Theta^{(s_{\nth})}_{\nth}
\end{equation}
where $\Theta_k^{(s_k)}=\left\{F_{\theta_k}^{-1}(y)\lvert y\in \Big(\frac{s_k-1}{\nstk{k}},\frac{s_k}{\nstk{k}}\Big]\right\}$
so that $\Prob{\thetaa\in\Thetaa^{(s)}}=\prod_{i=1}^{\nth}\left({1}/{\nstk{k}}\right)=1/\nst$.  


\begin{figure}
\centering
\begin{tikzpicture}[scale=0.75,every node/.style={minimum size=1cm},on grid,/tikz/fill between/on layer=main]
		  	
    \begin{scope}[
        yshift=-170,every node/.append style={
        yslant=0.5,xslant=-1},yslant=0.5,xslant=-1
                  ]


\draw [ultra thick,red,name path=cont1](3,0) parabola bend (4.5,1) (5,1)  ;
\draw [ultra thick,red!70,name path=cont2](2,0) parabola bend (4.5,2) (5,2)  ;
\draw [ultra thick,red!40,name path=cont3](1,0) parabola bend (4.5,3) (5,3)  ;
\draw [ultra thick,yellow,name path=cont4](0,0) parabola bend (4.5,4) (5,4)  ;
\draw [ultra thick,green!60,name path=cont5](0,1) parabola bend (4.5,5) (5,5)  ;
\draw [ultra thick,teal,name path=cont6](0,2) parabola bend (3,5) (3,5)  ;
\draw [ultra thick,cyan,name path=cont7](0,3) parabola bend (2,5) (2,5)  ;
\draw [ultra thick,blue,name path=cont8](0,4) parabola bend (1,5) (1,5)  ;

\draw [thick,name path=cont0](3,0) -- (5,0) -- (5,1);
\draw [thick,name path=cont9](0,4) -- (0,5) -- (1,5);

\tikzfillbetween[
    of=cont1 and cont2
  ] {pattern=grid, pattern color=red!60};
  \tikzfillbetween[
    of=cont2 and cont3
  ] {pattern=grid, pattern color=orange};
    \tikzfillbetween[
    of=cont0 and cont1
  ] {pattern=grid, pattern color=red};
      \tikzfillbetween[
    of=cont3 and cont4
  ] {pattern=grid, pattern color=yellow!70!orange};
      \tikzfillbetween[
    of=cont4 and cont5
  ] {pattern=grid, pattern color=green!60!yellow};
      \tikzfillbetween[
    of=cont5 and cont6
  ] {pattern=grid, pattern color=blue!20!green};
      \tikzfillbetween[
    of=cont6 and cont7
  ] {pattern=grid, pattern color=blue!50!green};
      \tikzfillbetween[
    of=cont7 and cont8
  ] {pattern=grid, pattern color=green!10!blue};
        \tikzfillbetween[
    of=cont8 and cont9
  ] {pattern=grid, pattern color=black!10!blue};

        \draw[black,very thick] (0,0) rectangle (5,5);
  
    \end{scope} 
    \begin{scope}[
            yshift=-83,every node/.append style={
            yslant=0.5,xslant=-1},yslant=0.5,xslant=-1
            ]
        \fill[white,fill opacity=0.75] (0,0) rectangle (5,5);
        \draw[step=10mm, black, thick] (0,0) grid (5,5); 
        \draw[black,very thick] (0,0) rectangle (5,5);
%
        
        \fill[yellow,fill opacity=0.15] (2,0) rectangle (5,2);
         \fill[yellow,fill opacity=0.15] (3,2) rectangle (5,3);
         \draw[orange!90!black, ultra thick] (2,0) -- (2,2) -- (3,2) -- (3,3) -- (5,3);

        \foreach \x in {0,...,4}
    \foreach \y in {0,...,4} 
    \foreach \z in {0,...,4}
    {\draw[red] (\x+0.5+0.5*rand,\y+0.5+0.5*rand) circle (0.025cm);
    }
    
            \foreach \x in {1,...,4}
    \foreach \y in {0,...,3} 
        \foreach \z in {0,...,10}
    {\draw[red] (\x+0.5+0.5*rand,\y+0.5+0.5*rand) circle (0.025cm);
    }
    
     \foreach \x in {2,...,4}
    \foreach \y in {0,...,2} 
        \foreach \z in {0,...,10}
    {\draw[red] (\x+0.5+0.5*rand,\y+0.5+0.5*rand) circle (0.025cm);
    }
    
         \foreach \x in {3,...,4}
    \foreach \y in {0,...,1} 
        \foreach \z in {0,...,50}
    {\draw[red] (\x+0.5+0.5*rand,\y+0.5+0.5*rand) circle (0.025cm);
    }
    
             \foreach \x in {4,...,4}
    \foreach \y in {0,...,0} 
        \foreach \z in {0,...,100}
    {\draw[red] (\x+0.5+0.5*rand,\y+0.5+0.5*rand) circle (0.025cm);
    }
    
                 \foreach \x in {2,...,4}
    \foreach \y in {4,...,4} 
        \foreach \z in {0,...,5}
    {\draw[red] (\x+0.5+0.5*rand,\y+0.5+0.5*rand) circle (0.025cm);
    }
    
                     \foreach \x in {0,...,0}
    \foreach \y in {0,...,3} 
        \foreach \z in {0,...,5}
    {\draw[red] (\x+0.5+0.5*rand,\y+0.5+0.5*rand) circle (0.025cm);
    }

    \end{scope}

           \draw[-latex,thick](2.3,2)node[right]{Stratum}
        to[out=180,in=90] (0,.5);

        
           \draw[-latex,thick](5.3,.75)node[right]{Samples}
        to[out=180,in=90] (1.85,-.5);

    \draw[-latex,thick](4,-5)node[right]{Iso-likelihood contours}
        to[out=180,in=90] (1.95,-4.75);

\contourlength{1.5pt}%
\node at (-4.15,-3.5) {\contour{white}{$\lambda_1$}};    
\node at (-3.40,-3.5) {\contour{white}{$\lambda_2$}};    
\node at (-2.65,-3.5) {\contour{white}{$\lambda_3$}};    

\node at (3.1,-3.5) {\contour{white}{$\lambda_n$}};    
\node at (2,-3.5) {\contour{white}{$\cdots$}};    

\node at (0.6,-3.5) {\contour{white}{$\lambda_i$}};   
\node at (-0.9,-3.5) {\contour{white}{$\cdots$}};   

 \fill[black]
      (-2.5,-5.8) node [above] {$\theta_1$}
         (2.4,-5.9) node [above] {$\theta_2$};	

\draw[fill = yellow!10,draw=orange!90!black,ultra thick] (6,-1) rectangle (6.75,-1.75); 
\node at (9.75,-1.4) {Strata with $\Le(\thetaa)>\lambda_i$};

\end{tikzpicture}
\caption{Likelihood level adapted stratified sampling (LLA-SS) method: the iso-likelihood contours are shown with $\lambda_1<\lambda_2<\dots<\lambda_n$; more samples are generated from the strata with high likelihood values. For example, the shaded strata with the thick boundary lines have at least one sample with $\Le(\thetaa)>\lambda_i$.} \label{fig:mlss}
\end{figure}
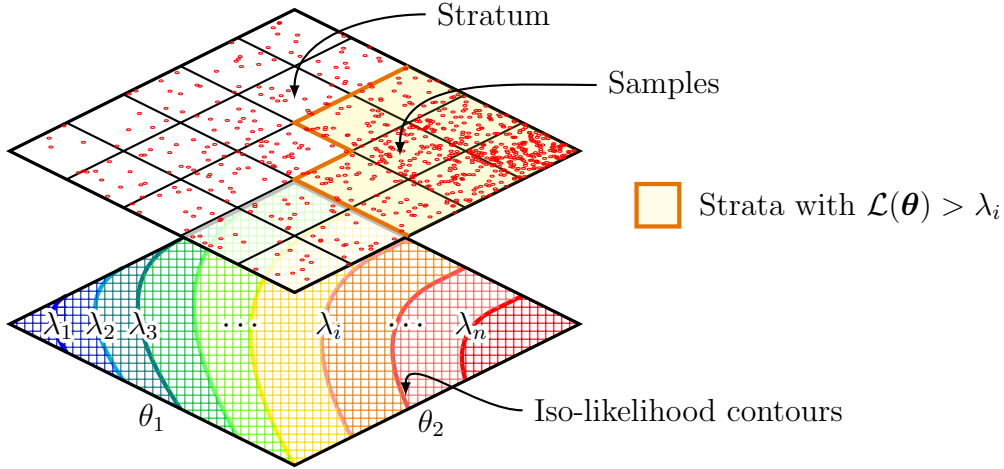

\subsection{Likelihood level adapted Markov chain Monte Carlo (LLA-MCMC)} \label{sec:lla-mcmc} 

In the following proposed algorithm, the Markov chain Monte Carlo method \cite{andrieu2003introduction} is used in a sequential manner. 
As the likelihood level $\lambda_i$ increases with iteration, the domain is reduced with $\widetilde\Thetaa_i \subset \widetilde\Thetaa_{i-1}$, where $\widetilde\Thetaa_i$ is the region in which the likelihood is above $\lambda_i$, as defined in \eqref{eq:Theta_def}. Hence, at iteration $i$, $\chii$ can be written as \cite{del2005genealogical,cerou2012sequential}
\begin{equation}\label{eq:MCMC_nested}
\begin{split}
\chii=\Prob{\thetaa\in\widetilde{\Thetaa}_i}&
=\prod_{k=1}^{i}\Prob{\thetaa\in\widetilde{\Thetaa}_k|\thetaa\in\widetilde{\Thetaa}_{k-1}}\\
&=\Prob{\thetaa\in\widetilde{\Thetaa}_i|\thetaa\in\widetilde{\Thetaa}_{i-1}} \prod_{k=1}^{i-1}\Prob{\thetaa\in\widetilde{\Thetaa}_k|\thetaa\in\widetilde{\Thetaa}_{k-1}}\\
&=\frac{\Prob{\thetaa\in\widetilde{\Thetaa}_i}}{\Prob{\thetaa\in\widetilde{\Thetaa}_{i-1}}} \widehatchisub{i-1}\\
&\approx \widehatchisub{i-1}\frac{\sum_{j=1}^{N}\left[\mathbb{I}_{\widetilde{\Thetaa}_i}\left(\thetaa_{j,i}\right)\right]}{\sum_{j=1}^N\left[\mathbb{I}_{\widetilde{\Thetaa}_{i-1}}\left(\thetaa_{j,i-1}\right)\right]} \\
\end{split}
\end{equation}
where $\widetilde{\Thetaa}_0\equiv \Thetaa$ and $\left\{\thetaa_{j,l} \right\}_{l>0}$, $j=1,\dots,N$, are obtained from Markov chains with $\thetaa_{j,0}\sim \pdf(\thetaa)$ with a transition kernel that can be assumed of the form $K(\theta,\mathrm{d}\eta)$ given in \eqref{eq:kernel}. Hence, starting with $\widehatchisub{0}=1$, $\widehatchisub{i}$ can be estimated in the $i$\textsuperscript{th} iteration by sampling from Markov chains that are initiated with the previous iteration's samples $\{\thetaa_{j,i-1}\}_{j=1}^N$, each satisfying $\Le\left(\thetaa_{j,i-1}\right)>\lambda_{i-1}$.
Note that, for a high-dimensional parameter space, a Markov chain with a higher acceptance rate is used herein, \eg the modified Metropolis-Hastings algorithm (MMH) \cite{au2001estimation}. 
An algorithm for likelihood-level adapted estimation of marginal likelihood or evidence using a Markov chain is presented in Algorithm \ref{alg:mlpa}, where $\pdf_i(\thetaa)=\pdf(\thetaa|\thetaa\in\widetilde{\Thetaa}_i)$ 
are successively defined and samples are drawn from them using the Markov chain Monte Carlo method. A schematic of this algorithm, with the Markov chain propagating samples to higher likelihood regions, is shown in \fref{fig:mlpa}. In the figure, Markov chains are shown to propagate to high-likelihood regions as samples are drawn from them.

\begin{algorithm}
	\DontPrintSemicolon
	\CommentSty{\color{newblue}}
	{\bf{Initialization:}} Set $\widehatchisub{0}^{~\!\mathrm{MCMC}}=1$, $\widehat\Evd^{~\!\mathrm{MCMC}}_0=0$, $\lambda_0=0$, and $\pdf_0(\thetaa)=\pdf(\thetaa)$\;
	Draw $N$ samples $\{\thetaa_{j,0}\}_{j=1}^N$ from the distribution $\pdf_0(\thetaa)$\;
	Set $i=1$\;
	\While{Stopping Criterion = FALSE}{
	Select a suitable new likelihood level $\lambda_i>\lambda_{i-1}$ (see Section~\ref{sec:lik_level})\;
		Define $\pdf_i(\thetaa)=\pdf(\thetaa|\thetaa\in\widetilde\Thetaa_i)$\;
		Identify the $N_\mathrm{pass} \leq N$ passing samples $\{\thetaa_{k,i}\}_{k=1}^{N_\mathrm{pass}} = \{\thetaa_{j,i} | \Le(\thetaa_{j,i})>\lambda_i \}$\;
        Remove the $N-N_\mathrm{pass}$ non-passing samples for which $\Le(\thetaa_{j,i})\leq\lambda_i$\;
		Draw $N-N_\mathrm{pass}$ samples from the distribution $\pdf_i(\thetaa)$ by running $N-N_\mathrm{pass}$ Markov chains, each starting from a sample randomly selected from $\{\thetaa_{k,i}\}_{k=1}^{N_\mathrm{pass}}$ to replace a non-passing sample\; 
		Evaluate likelihood values $\Le\left(\thetaa_{j,i}\right)$ for the newly drawn samples\; 
		
		Estimate $\widehatchisub{i}^{~\!\mathrm{MCMC}}= {\widehatchisub{i-1}^{~\!\mathrm{MCMC}}}\sum_{j=1}^{N}\mathbb{I}_{\widetilde{\Thetaa}_i}(\thetaa_{j,i})/N$\;
		Update the marginal likelihood $\widehat\Evd^{~\!\mathrm{MCMC}}_i=\widehat\Evd^{~\!\mathrm{MCMC}}_{i-1}+\lambda_i\left(\widehatchisub{i-1}^{~\!\mathrm{MCMC}}-\widehatchisub{i}^{~\!\mathrm{MCMC}}\right)$\;
		$i\leftarrow i+1$\;
	}
	\KwResult{The marginal likelihood $\widehat\Evd^{~\!\mathrm{MCMC}}\subscript{final}$}
	%
	%
	\vspace{2pt}
	\caption{Likelihood level adapted marginal likelihood estimation using particle approximation (LLA-MCMC).}\label{alg:mlpa}
\end{algorithm}

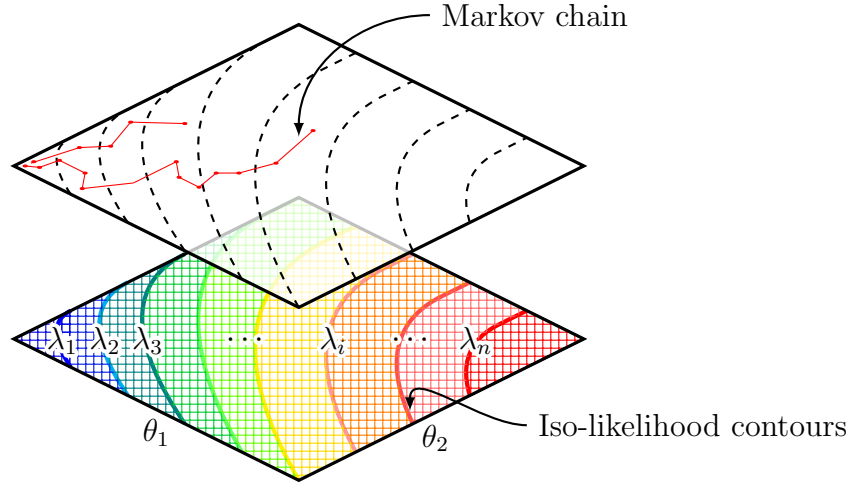
\begin{figure}
\centering
\begin{tikzpicture}[scale=0.75,every node/.style={minimum size=1cm},on grid,/tikz/fill between/on layer=main]
		  	
    \begin{scope}[
        yshift=-170,every node/.append style={
        yslant=0.5,xslant=-1},yslant=0.5,xslant=-1
                  ]


\draw [ultra thick,red,name path=cont1](3,0) parabola bend (4.5,1) (5,1)  ;
\draw [ultra thick,red!70,name path=cont2](2,0) parabola bend (4.5,2) (5,2)  ;
\draw [ultra thick,red!40,name path=cont3](1,0) parabola bend (4.5,3) (5,3)  ;
\draw [ultra thick,yellow,name path=cont4](0,0) parabola bend (4.5,4) (5,4)  ;
\draw [ultra thick,green!60,name path=cont5](0,1) parabola bend (4.5,5) (5,5)  ;
\draw [ultra thick,teal,name path=cont6](0,2) parabola bend (3,5) (3,5)  ;
\draw [ultra thick,cyan,name path=cont7](0,3) parabola bend (2,5) (2,5)  ;
\draw [ultra thick,blue,name path=cont8](0,4) parabola bend (1,5) (1,5)  ;

\draw [thick,name path=cont0](3,0) -- (5,0) -- (5,1);
\draw [thick,name path=cont9](0,4) -- (0,5) -- (1,5);

\tikzfillbetween[
    of=cont1 and cont2
  ] {pattern=grid, pattern color=red!60};
  \tikzfillbetween[
    of=cont2 and cont3
  ] {pattern=grid, pattern color=orange};
    \tikzfillbetween[
    of=cont0 and cont1
  ] {pattern=grid, pattern color=red};
      \tikzfillbetween[
    of=cont3 and cont4
  ] {pattern=grid, pattern color=yellow!70!orange};
      \tikzfillbetween[
    of=cont4 and cont5
  ] {pattern=grid, pattern color=green!60!yellow};
      \tikzfillbetween[
    of=cont5 and cont6
  ] {pattern=grid, pattern color=blue!20!green};
      \tikzfillbetween[
    of=cont6 and cont7
  ] {pattern=grid, pattern color=blue!50!green};
      \tikzfillbetween[
    of=cont7 and cont8
  ] {pattern=grid, pattern color=green!10!blue};
        \tikzfillbetween[
    of=cont8 and cont9
  ] {pattern=grid, pattern color=black!10!blue};

        \draw[black,very thick] (0,0) rectangle (5,5);
  
    \end{scope} 
    \begin{scope}[
            yshift=-83,every node/.append style={
            yslant=0.5,xslant=-1},yslant=0.5,xslant=-1
            ]

        \fill[white,fill opacity=0.75] (0,0) rectangle (5,5);
        
        \draw [thick,name path=cont1,dashed](3,0) parabola bend (4.5,1) (5,1)  ;
\draw [thick,name path=cont2,dashed](2,0) parabola bend (4.5,2) (5,2)  ;
\draw [thick,name path=cont3,dashed](1,0) parabola bend (4.5,3) (5,3)  ;
\draw [thick,name path=cont4,dashed](0,0) parabola bend (4.5,4) (5,4)  ;
\draw [thick,name path=cont5,dashed](0,1) parabola bend (4.5,5) (5,5)  ;
\draw [thick,name path=cont6,dashed](0,2) parabola bend (3,5) (3,5)  ;
\draw [thick,name path=cont7,dashed](0,3) parabola bend (2,5) (2,5)  ;
\draw [thick,name path=cont8,dashed](0,4) parabola bend (1,5) (1,5)  ;
        \draw[black,very thick] (0,0) rectangle (5,5);
%

\draw[red,fill=red] (0.1,4.9) circle (0.025cm);
\draw[red] (0.1,4.9) -- (0.2,4.75);
\draw[red,fill=red] (0.2,4.75) circle (0.025cm);
\draw[red] (0.5,4.7) -- (0.2,4.75);

\draw[red,fill=red] (0.5,4.7) circle (0.025cm);

\draw[red] (0.5,4.7) -- (0.5,4.25);

\draw[red,fill=red] (0.5,4.25) circle (0.025cm);
\draw[red] (0.2,4) -- (0.5,4.25);

\draw[red,fill=red] (0.2,4) circle (0.025cm);
\draw[red] (0.2,4) -- (0.75,3.65);

\draw[red,fill=red] (1.5,3.65) circle (0.025cm);
\draw[red] (1.5,3.65) -- (0.75,3.65);

\draw[red,fill=red] (1.25,3.35) circle (0.025cm);
\draw[red] (1.25,3.35) -- (1.5,3.65);

\draw[red,fill=red] (1.25,3) circle (0.025cm);
\draw[red] (1.25,3) -- (1.25,3.35);

\draw[red,fill=red] (1.65,3.1) circle (0.025cm);
\draw[red] (1.25,3) -- (1.65,3.1);

\draw[red,fill=red] (1.85,2.9) circle (0.025cm);
\draw[red] (1.65,3.1) -- (1.85,2.9);

\draw[red,fill=red] (2.35,2.75) circle (0.025cm);
\draw[red] (1.85,2.9) -- (2.35,2.75);

\draw[red,fill=red] (3.25,3) circle (0.025cm);
\draw[red] (3.25,3) -- (2.35,2.75);

\draw[red,fill=red] (0.25,4.9) circle (0.025cm);
\draw[red] (0.2,4.9) -- (0.9,4.75);
\draw[red,fill=red] (0.9,4.75) circle (0.025cm);
\draw[red] (1.2,4.5) -- (0.9,4.75);
\draw[red,fill=red] (1.2,4.5) circle (0.025cm);
\draw[red] (1.2,4.5) -- (1.8,4.75);
\draw[red,fill=red] (1.8,4.75) circle (0.025cm);
\draw[red] (1.8,4.75) -- (2.25,4.25);
\draw[red,fill=red] (2.25,4.25) circle (0.025cm);

%
%
%
%
%

    \end{scope}

           \draw[-latex,thick](2.3,2.25)node[right]{Markov chain}
        to[out=180,in=90] (0,.1);

        

    \draw[-latex,thick](4,-5)node[right]{Iso-likelihood contours}
        to[out=180,in=90] (1.95,-4.75);

\contourlength{1.5pt}%
\node at (-4.15,-3.5) {\contour{white}{$\lambda_1$}};    
\node at (-3.40,-3.5) {\contour{white}{$\lambda_2$}};    
\node at (-2.65,-3.5) {\contour{white}{$\lambda_3$}};    

\node at (3.1,-3.5) {\contour{white}{$\lambda_n$}};    
\node at (2,-3.5) {\contour{white}{$\cdots$}};    

\node at (0.6,-3.5) {\contour{white}{$\lambda_i$}};   
\node at (-0.9,-3.5) {\contour{white}{$\cdots$}};

 \fill[black]
      (-2.5,-5.8) node [above] {$\theta_1$}
         (2.4,-5.9) node [above] {$\theta_2$};

\end{tikzpicture}
\caption{Likelihood level adapted particle approximation (LLA-MCMC) method: the iso-likelihood contours are shown with $\lambda_1<\lambda_2<\dots<\lambda_n$; Markov chains are run to generate samples with $\Le(\thetaa)>\lambda_i$ starting from a sample with $\Le(\thetaa)>\lambda_{i-1}$.} \label{fig:mlpa}
\end{figure}


%



\section{Discussion of the Proposed Approach}\label{sec:disc}

\subsection{Estimation of posterior moments}
The posterior statistics of the model parameters are often sought from a Bayesian analysis.
Using the approach proposed herein, the samples corresponding to each of the likelihood levels from the three algorithms may also be used to evaluate the posterior moments of the model parameters $\thetaa$, using the rejected samples and the change in evidence value at each step, without any significant computational cost. For example, the mean and variance of $\thetaa$ can be evaluated using
\begin{equation}
\begin{split}
\Exp[\thetaa]&\approx\frac{1}{\widehat\Evd\subscript{final}}\sum_{i=1}^{i\subscript{final}}\Delta\widehat\Evd_i \bar\thetaa_{i}\\
\mathrm{Var}[\thetaa]&\approx\frac{1}{\widehat\Evd\subscript{final}}\sum_{i=1}^{i\subscript{final}}\Delta\widehat\Evd_i \bar\thetaa_{i}^2-\left(\Exp[\thetaa]\right)^2\\
\end{split}
\end{equation}
where $\Delta\widehat\Evd_i=\lambda_i(\widehatchisub{i-1}-\widehatchisub{i})$ and $\bar\thetaa_{i}$ is either a randomly chosen sample from (which avoids estimating the parameter means at every level), or the mean of (which is more accurate), $\{\thetaa_{j,i}|\Le(\thetaa_{j,i})\in[\lambda_{i-1},\lambda_i]\}$ that has been used to estimate $\chii$.


\subsection{Stopping criteria}\label{sec:stop}
Different stopping criteria, based on accuracy and/or computational cost, can be used in these algorithms, namely: (i) the change in evidence is less than 
some threshold $\Delta\Evd_{\mathrm{tol}}$, often taken to be 1\% or 0.1\%;
(ii) the total number of iterations exceeds a pre-chosen number $i_\mathrm{final}$; (iii) $\widehatchisub{i}$ is less than some pre-specified tolerance $\chisub{\mathrm{tol}}$; and/or (iv) the number of sample likelihood calculations has exceeded some threshold.
The examples herein utilize one or more of these stopping criteria.  While not used herein, an additional criterion could be to stop when the likelihood level $\lambda_i$ of the current iteration is within some fraction of the theoretical maximum of the likelihood function (if known).

\subsection{Accuracy}
If $\chisub{\mathrm{final}} \equiv \chisub{i_\mathrm{final}}$ is the final (\ie after iteration $i_\mathrm{final}$) probability mass contained in regions with likelihoods greater than $\lambda_\mathrm{final} \equiv \lambda_{i_\mathrm{final}}$, the evidence in \eqref{eq:evidencenested} can be written as \cite{chopin2010properties}
\begin{equation}
\begin{split}
\Evd&=\int_0^1 \varphi(\chi)d\chi\\
&=\sum_{i=1}^{i_\mathrm{final}} \lambda_i \Delta\widehatchisub{i}+\underbrace{\int_{0}^{\chisub{\mathrm{final}}} \varphi(\chi)d\chi }_{\errorcomponent_\mathrm{t}}+ \underbrace{\left[\int_{\chisub{\mathrm{final}}}^{1} \varphi(\chi)d\chi-\sum_{i=1}^{i_\mathrm{final}} \lambda_i \Delta\chi(\lambda_i)\right]}_{\errorcomponent_\mathrm{n}}\\
&\qquad\qquad \qquad\qquad +\underbrace{\sum_{i=1}^{i_\mathrm{final}}\lambda_i\left[\Delta\chi(\lambda_i)-\Delta\widehatchisub{i}\right]}_{\errorcomponent_\mathrm{s}}\\
&=\sum_{i=1}^{i_\mathrm{final}} \lambda_i \Delta\widehatchisub{i}+\errorcomponent_{\mathrm{t}}+\errorcomponent_{\mathrm{n}}+\errorcomponent_{\mathrm{s}}
\end{split}
\end{equation}
where $\Delta \widehat \chii = (\widehatchisub{i-1}-\widehatchisub{i})$ is an estimate of $\Delta\chi(\lambda_i)=\chi(\lambda_{i-1})-\chi(\lambda_i)$, $\errorcomponent_{\mathrm{t}}$ is the truncation error, $\errorcomponent_\mathrm{n}$ is the numerical integration error, and $\errorcomponent_\mathrm{s}$ is the stochastic error. 

Since the algorithms are halted when $\chi=\chisub{\mathrm{final}} \ll 1$, rather than at $\chi=0$, the truncation error $\errorcomponent_{\mathrm{t}}$ arises; however, $\errorcomponent_\mathrm{t}$ is very small if a sufficient number of iterations are used. 
If $\mathrm{d}\varphi/\mathrm{d}\chi$ is bounded in $[\chisub{\mathrm{final}},1]$, the numerical integration error $\errorcomponent_{\mathrm{n}}$ in the proposed algorithms for a rectangular rule is of the order $\mathcal{O}(N^{-1})$, where $N$ is the number of uncorrelated samples. Finally, the stochastic error $\errorcomponent_{\mathrm{s}}$ is asymptotically unbiased with a convergence rate of $\mathcal{O}(N^{-1/2})$ for the methods used here, as shown in \ref{sec:app1}. Hence, the convergence of $\errorcomponent_\mathrm{s}$ dictates the convergence of the algorithms here due to its slower dependence on $N^{-1/2}$.

\subsection{Choice of likelihood levels}\label{sec:lik_level}
The numerical integration error $\errorcomponent_\mathrm{n}$ can be reduced by using very small decrements in $\chi$, which, however, leads to higher computational cost. 
On the other hand, choosing $\chii \ll \chisub{i-1}$ will lead to an inaccurate estimate of the marginal likelihood. 
Note that the decrement in $\chii$ depends on the selection and increment of $\lambda_i$. Hence, to ensure that the decrements in $\chi$ are not very large, small increments in $\lambda$ are suggested.
For example, based on the computational budget and required accuracy, $\lambda_i$ can be assumed of the form $\lambda_i=10^\alpha\lambda_{i-1}$ with $\alpha>0$ chosen accordingly. 
Another strategy, utilized in the numerical examples herein, is to choose $\lambda_i$ as the $N^\mathrm{reject}_i$th lowest of all likelihood values
, where $N^\mathrm{reject}_i=\lceil{f_i^\mathrm{reject}N_i}\rceil$ for some fixed $f_i^\mathrm{reject}\in(0,1)$. Note that, at the first few iterations, a fixed but large $f_i^\mathrm{reject}$ may result in $\lambda_i<\lambda_{i-1}$; to avoid this herein, $f_i^\mathrm{reject}$ is gradually increased during the initial iterations to a fixed maximum value and, if $\lambda_i<\lambda_{i-1}$ is observed, $f_i^\mathrm{reject}$ is increased to avoid any such issues. 

\section{Numerical Illustrations}\label{sec:ex} 
The three proposed methods for evaluating the evidence are compared against other established approaches through a series of four examples in this section.  While the first example has an exact solution for the marginal likelihood, a standard Monte Carlo-based sampling approach is treated as the ground truth evidence value in the remaining three examples.  Examples I and II present comparisons against MultiNest, whereas Examples II, III, and IV provide comparisons against nested sampling.  
A high-level comparison of some of the characteristics of the algorithms implemented within this study is given in Table~\ref{tab:alg_comp}. These features include the capability for implementation in parallel computation (parallelizability), how well the method scales to high-dimensional problems (scalability), the ability to handle multi-modal likelihood functions (multi-modal), and the number of hyper-parameters that must be tuned (implementation complexity).

Table~\ref{tab:alg_comp} shows that the standard Monte Carlo approach is (embarrassingly) parallelizable and has only one hyper-parameter, \ie the number of samples, but it is significantly inefficient in higher dimensions and for multi-modal distributions.  
As mentioned previously, traditional nested sampling struggles with efficiency, whereas MultiNest is well-suited to handle multi-modal distributions.   
Traditional nested sampling is more efficient for high-dimensional spaces compared to MultiNest as noted in Example IV, though recent developments in MultiNest have enabled a certain degree of parallelization \cite{Dittmann2024}. 
Both nested sampling and MultiNest require specifying the number of samples (active points) 
and the stopping criterion, which is based on either the total number of iterations or achieving some tolerance related to the change in evidence.  In addition, MultiNest has an additional hyper-parameter for the target efficiency. 
The hyperparameters chosen herein for MultiNest are consistent with those used for nested sampling; the hyperparameters for the three proposed algorithms are chosen to provide a comparable number of function evaluations for a given tolerance. 


Among the three proposed methodologies, LLA-IS lacks the capability for parallelization and is not well-suited to explore high-dimensional space.  LLA-IS also features three hyper-parameters, including the number of initial samples and the standard deviation of the importance density, both of which can be tuned with respect to the available computational power and model complexity. Both the LLA-IS and LLA-SS methods feature a hyper-parameter related to the ``fraction of likelihood level,'' which must be tuned as described in Section~\ref{sec:lik_level}. The advantage of LLA-SS is its fast implementation (thanks to parallelization) and high accuracy for low-dimensional space. The hyper-parameter for the number of substrata for each parameter can also be fixed ($\nstk{k}=5$ is used herein). 
The proposed LLA-MCMC approach can also be applied in parallel; its strength, however, is in its scalability, 
as is powerfully demonstrated in Example IV. Its stopping criterion hyper-parameter can be determined based on Section~\ref{sec:stop}. 
The other three hyper-parameters are selected based on model complexity and the available computational power. 
Like standard Monte Carlo sampling, the LLA-SS algorithm can handle multi-modal likelihood functions, but not as well as MultiNest.
The LLA-IS and LLA-MCMC algorithms can also be used for such likelihood functions, but LLA-IS requires a carefully chosen importance density function $\pdfis(\theta)$, and LLA-MCMC requires further modifications \cite{latuszynski2025mcmc} (that replace the MH algorithm used herein with different MCMC algorithms). 


\begin{table}
    \centering
    \caption{Comparison of the computational capabilities of the algorithms implemented in this study} 
    
    \label{tab:alg_comp}
    \resizebox{\columnwidth}{!}{
    \begin{tabular}{@{}l|c|c|c|l@{}}
        \hline
        \HeaderTstrut\HeaderBstrut \textbf{Algorithm} & \textbf{Parallelizability} & \textbf{Scalability} & \textbf{Multi-Modal} & \textbf{Hyperparameters} \\
        \hline
        \Tstrut\Bstrut Monte Carlo & Perfectly & Inefficiently & Yes & {\footnotesize1. \# of samples} \\ \hline
        \Tstrut\Bstrut Nested Sampling & No & Inefficiently & No & {\footnotesize\begin{tabular}[c]{@{}l@{}}1. \# of samples\\ 2. Stopping criterion \end{tabular}} \\ \hline
        \Tstrut\Bstrut MultiNest & Yes & No & Well-suited & {\footnotesize\begin{tabular}[c]{@{}l@{}}1. \# of samples (active points)\\ 2. Target efficiency\\ 3. Stopping criterion \end{tabular}} \\ \hline
        \Tstrut\Bstrut LLA-IS & No & No & \begin{tabular}[c]{@{}c@{}}Depends on\\ISD choice\end{tabular} & {\footnotesize\begin{tabular}[c]{@{}l@{}}1. \# of initial samples\\ 2. Gaussian ISD std.\\ 3. Fraction of likelihood level\end{tabular}} \\ \hline
        \Tstrut\Bstrut LLA-SS & Yes & No & Yes & {\footnotesize\begin{tabular}[c]{@{}l@{}}1. \# of samples/strata\\ 2. \# of strata\\ 3. Fraction of likelihood level\end{tabular}} \\ \hline
        \Tstrut\Bstrut LLA-MCMC        & Yes               & Efficiently   & \begin{tabular}[c]{@{}c@{}}With \\ modifications\end{tabular} & {\footnotesize\begin{tabular}[c]{@{}l@{}}1. \# of samples\\ 2. \# of samples being replaced\\ 3. Gaussian proposal density std.\\ 4. Stopping criterion\end{tabular}}\\
        \hline
    \end{tabular}
    }
\end{table}

\subsection{Example I: Conceptual example}
Consider a simple multivariate Gaussian likelihood function for a dataset $\D=\{x_k\}_{k=1}^n$ defined by
\begin{equation}
\Le(\mu)=\pdf(\D|~\mu)=\frac{1}{(2\pi)^{n/2}\sigma^{n}}\exp\left(-\frac{1}{2\sigma^{2}}\sum_{k=1}^n[x_k-\mu]^2\right)
\end{equation}
The prior for the parameter $\mu$ is assumed to be a normal distribution with mean and variance given by $\mu_0$ and $\sigma_0^2$, respectively, \ie
\begin{equation}
\pdf(~\mu~|~\mu_0,\sigma_0^2)=\frac{1}{\sqrt{2\pi\sigma_0^2}}\exp\left(-\frac{(~\mu-\mu_0)^2}{2\sigma_0^{2}}\right)
\end{equation}
On the other hand, the standard deviation $\sigma$ of the likelihood function is assumed constant and known.
The evidence or marginal likelihood in this example can be easily computed~\cite{murphy2007conjugate}
\begin{equation}\label{eq:exact}
\begin{split}
\Evd_\mathrm{exact}=\frac{\sigma}{\left(\sqrt{2\pi\sigma^2}\right)^n\sqrt{n\sigma_0^2+\sigma^2}}\exp\left(-\frac{\sum_{k=1}^{n}x_k^2}{2\sigma^2}-\frac{\mu_0^2}{2\sigma_0^2}
+
\frac{2n\mu_0\bar{x}+\frac{\sigma_0^2n^2\bar{x}^2}{\sigma^2}+\frac{\sigma^2\mu_0^2}{\sigma_0^2}}{2(n\sigma_0^2+\sigma^2)}\right)
\end{split}
\end{equation}

With $n=100$ measurements generated from a normal distribution with mean $1.5$ and standard deviation $0.5$, the three proposed algorithms are implemented using $\mu_0=1$ and $\sigma_0=0.25$. The stopping criteria, as discussed in section \ref{sec:stop}, use evidence change threshold $\Delta\Evd_\mathrm{tol}=0.01\%$, lower bound tolerance $\chisub{\mathrm{tol}}=0.005$, maximum iteration count $i\subscript{max}=100$, and 
the number of likelihood evaluations is limited to 20,000. 

The MultiNest algorithm is implemented using $N=500$ samples. 
LLA-IS is implemented with a Gaussian importance density formed at each iteration using the mean of the retained samples and a smaller standard deviation of 0.125 (\ie a smaller spread) with an initial sample size of $N=1000$. The likelihood levels 
are decided at every iteration based on the fraction $f_i^\mathrm{reject}=\min(0.3,0.025i)$.
The LLA-SS algorithm is used with $\Theta^{(s)}=\Theta^{(s_1)}=\left\{F^{-1}_{\mu}\left(y\right)|y \in \Big(\frac{s_1-1}{\nstk{1}},\frac{s_1}{\nstk{1}}\Big]\right\}$, where $\nst=\nstk{1}=5$, and $F_\mu$ is the cumulative probability distribution of the parameter $\mu$. For this algorithm, the likelihood levels are decided at every iteration based on the fraction $f_i^\mathrm{reject}=\min(0.9,0.025i)$.
The LLA-MCMC algorithm is implemented with an initial sample size of $N=1000$ and, at each iteration, the 25 samples with the lowest likelihoods are rejected and 25 new samples with higher likelihoods are added to the sample pool. 

Table~\ref{tab:ex1_evd} shows a comparison of the marginal likelihoods or evidence values obtained using the three proposed algorithms, MultiNest, and the exact value computed using \eqref{eq:exact}. Note that the (natural) log of the evidence is used in reporting the results herein as the evidence can be orders of magnitude larger or smaller depending on the problem. The coefficient of variation (COV) is also used in this example to facilitate comparisons across multiple runs.
The table shows that all three proposed algorithms provide log evidences that are accurate within $0.4\%$ and exhibit very small coefficients of variation. Notably, the proposed LLA-SS outperforms MultiNest in terms of both accuracy and consistency via its lower error and COV, respectively.  
\fref{fig:ml_pa_error} shows how these algorithms reduce error with increasing computational effort. The figure also indicates that the LLA-SS algorithm outperforms the others (thought it will suffer from the curse of dimensionality as the parameter space dimension increases). The figure also shows that the COV of LLA-MCMC is much larger than those of the other algorithms because the samples from Markov chain Monte Carlo methods are generally correlated; this increased variation can be ameliorated by selecting every $n$\textsuperscript{th} sample from the chain, where $n\gg 1$.

\begin{figure}[htb!]
	\centering
	\begin{subfigure}[t]{0.475\textwidth}
		\centering
		\includegraphics[scale=0.25]{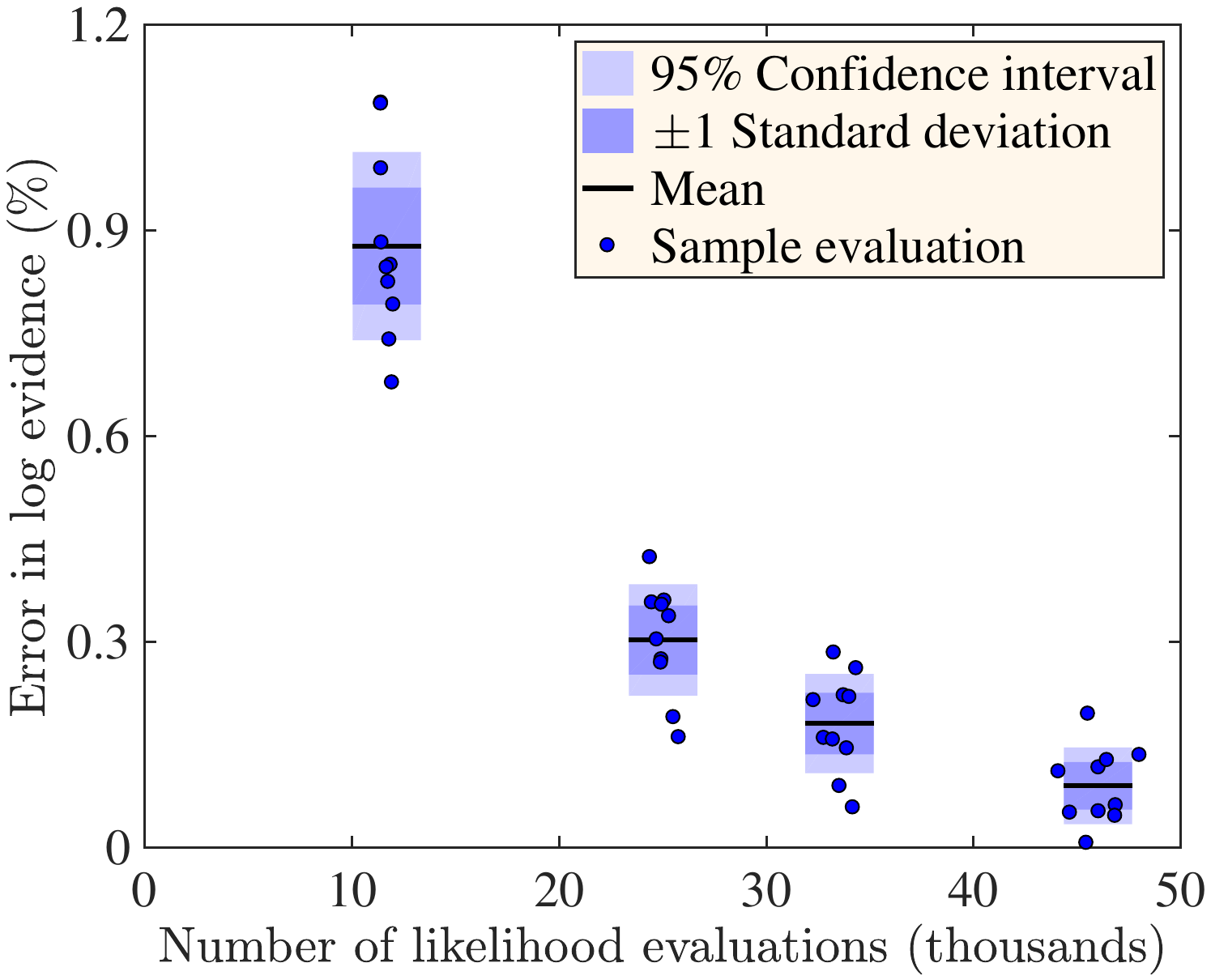}
		\caption{LLA-IS}
	\end{subfigure}
\hfill
	\begin{subfigure}[t]{0.475\textwidth}
		\centering
		\includegraphics[scale=0.25]{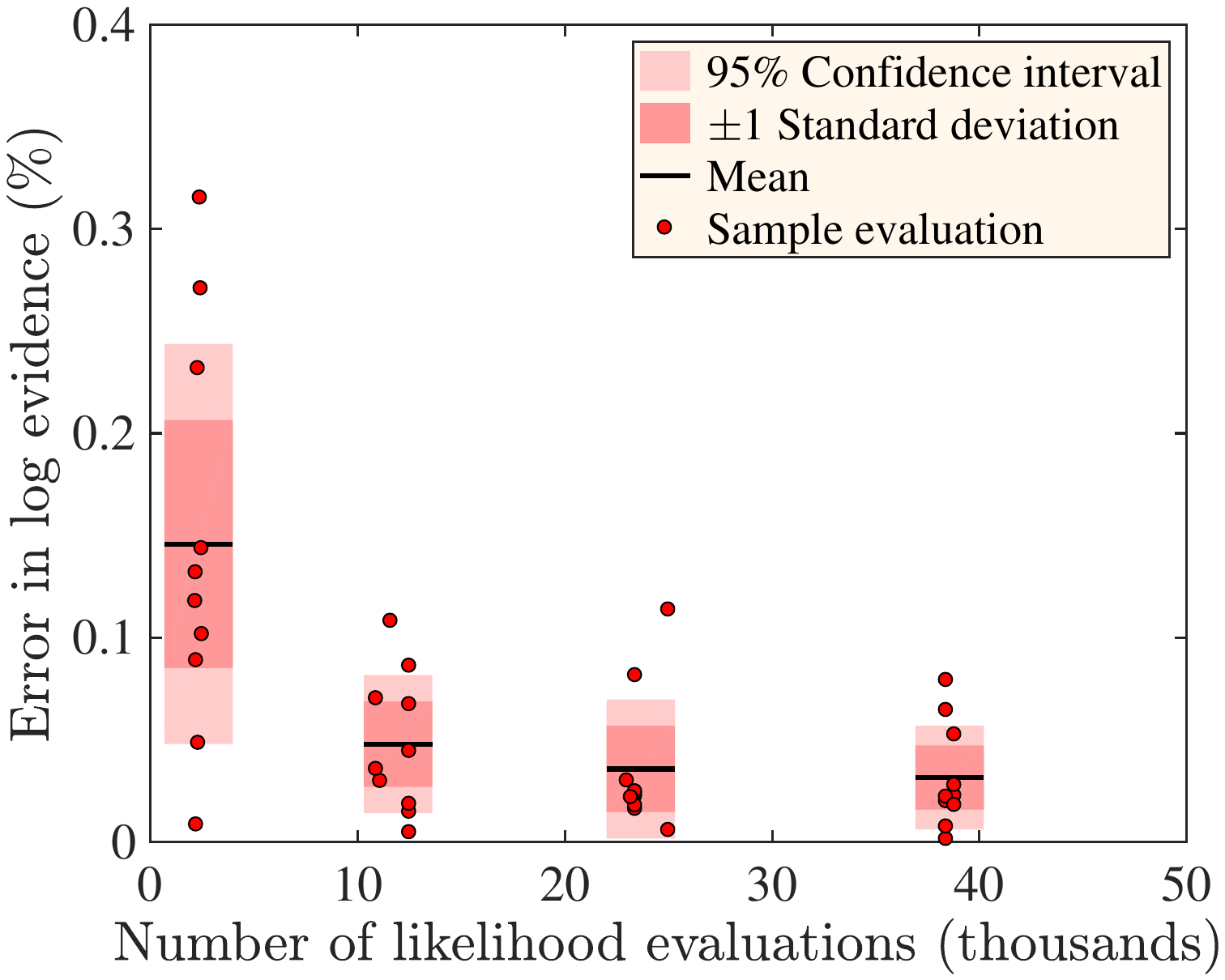}
		\caption{LLA-SS}
	\end{subfigure}
	\begin{subfigure}[t]{0.475\textwidth}
	\centering
	\includegraphics[scale=0.25]{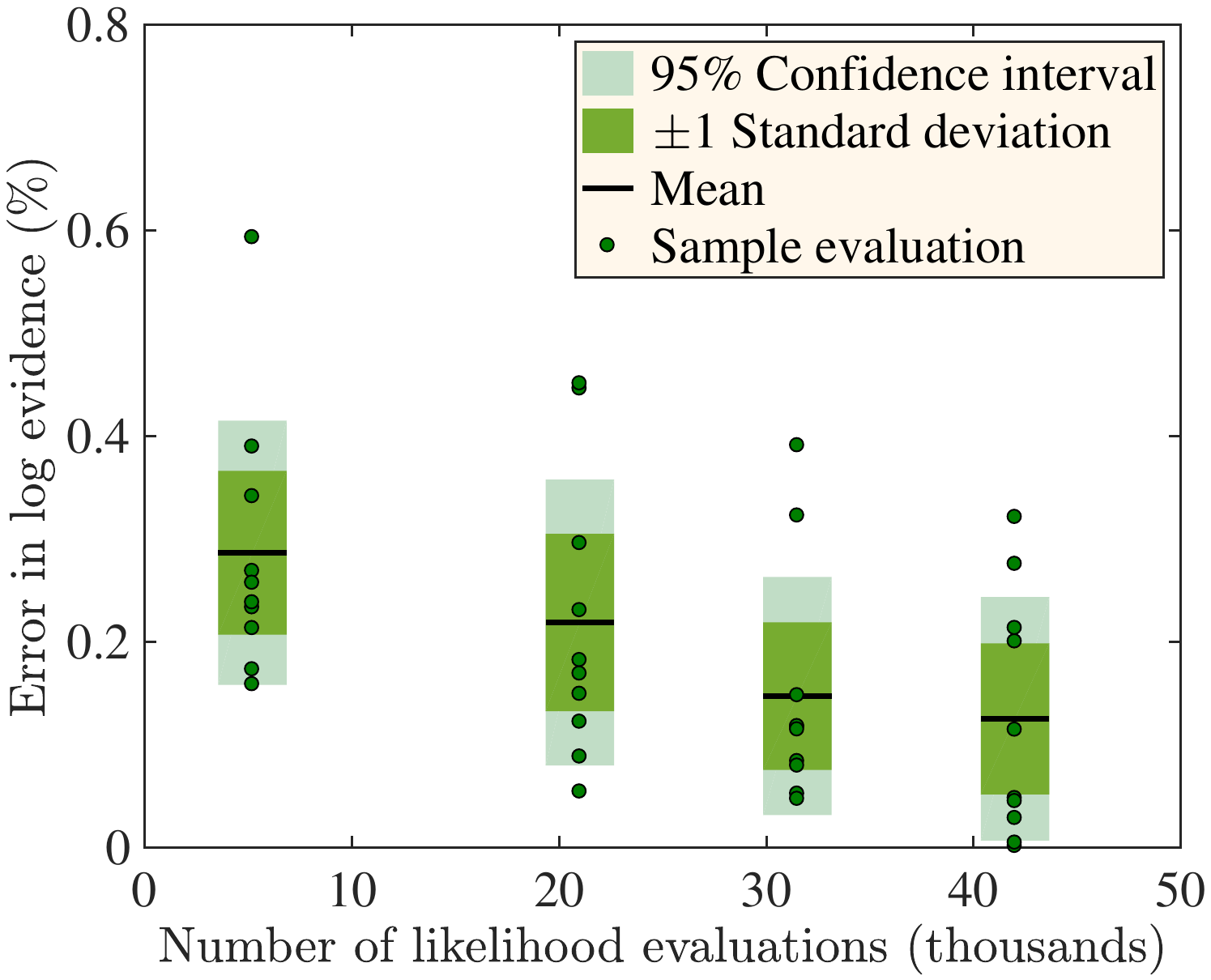}
	\caption{LLA-MCMC}
\end{subfigure}
	\caption{The error in the estimation of log evidence is reduced with increasing numbers of likelihood function evaluations for the three proposed algorithms. Note that the three vertical scales are not the same.} \label{fig:ml_pa_error}\vspace{-5pt}
\end{figure}

\begin{table}[htb!]
	\vspace{0pt}
	\centering
	\caption{Comparison of the first example's marginal likelihoods obtained using the three proposed algorithm along with the exact value. The errors and coefficients of variations (COV) of the log evidence are obtained from 10 independent simulation runs. } \label{tab:ex1_evd}
	\vspace{0pt}
	\begin{threeparttable}
	\begin{tabular}{@{} l | c| c| c|c @{}} 
		\hline
		\HeaderTstrut\HeaderBstrut
		\textbf{Method} & \textbf{\# fcn.~evals}\tnote{$\dagger$} & \textbf{log\textit{\textsubscript{e}}(Evidence)} & \textbf{Error (\%)} & \textbf{COV (\%)} \\
		\hline\Tstrut
		Exact & --- & $-$69.8354 & ---& ---\\
		MultiNest  & $\sim$ 10,000 & $-$69.8578 & 0.0751 & 0.0808 \\
		LLA-IS & $\sim$ 12,000 & $-$70.0468 & 0.3028 &0.0812\\
        {LLA-SS} & $\sim$ 10,000 & $-$69.8274 & 0.0113 & 0.0415\\
		\Bstrut LLA-MCMC & $\sim$ 10,500 & $-$69.7469 & 0.1267 & 0.1188\\
		\hline
	\end{tabular}
		\begin{tablenotes}
            \item[$\dagger$] {\small Since the number of function evaluations slightly varies for multiple runs, an approximate order-of-magnitude average is provided.}
		\end{tablenotes}
	\end{threeparttable}
\end{table}

%

\subsection{Example II: 11 story base isolated building}
\begin{figure}[htb!]
	\centerline{
		\includegraphics[scale=0.6]{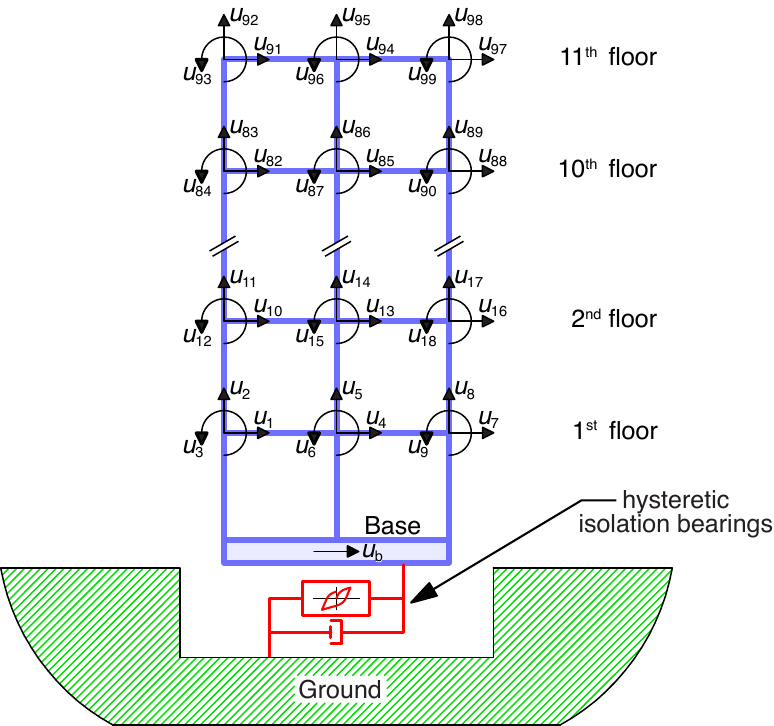}
	} \caption{100 DOF base-isolated structural model} \label{fig:building}\vspace{-5pt}
\end{figure}

This example utilizes a superstructure comprised of an 11-story 2-bay 99 DOF superstructure \cite{kamalzare2014computationally}, sitting on a hysteretic isolation layer that is rigid in-plane and can only move horizontally, resulting in the 100 DOF model shown in \fref{fig:building}.
%
Rayleigh damping, with assumed 3\% damping ratios for the 1\textsuperscript{st} and 10\textsuperscript{th} modes, is used to construct the superstructure damping matrix. 
If considered as a fixed base structure, the superstructure would have a fundamental period of 1.05\units{s} with  
equations of motion written as
\begin{equation}
\Mm\subs \ddot{\uu}\subs+\Cm\subs \dot{\uu}\subs +\Km\subs \uu\subs=-\Mm\subs \rr\subs \ddot u\subg
\end{equation}
where $\Mm\subs$, $\Cm\subs$ and $\Km\subs$ are the mass, damping  and stiffness matrices of the superstructure, respectively; $\ddot u\subg$ is the ground acceleration; $\uu\subs$ is the generalized displacement vector relative to the ground, consisting of 3 DOF per node for each of the 33 nodes in the superstructure. The influence vector of the ground motion $\ddot{u}\subg$, consisting of a 1 in each element corresponding to a horizontal displacement in $\uu\subs$ and zeros elsewhere, is denoted by $\rr\subs$. When combined with the isolation layer, the equations of motion are
\begin{gather}
\Mm\subs \ddot{\uu}\subs+\Cm\subs \dot{\uu}\subs +\Km\subs \uu\subs=-\Mm\subs \rr\subs \ddot u\subg+\Cm\subs\rr\subs \dot{u}\subb+\Km\subs \rr\subs u\subb
\\
m\subb\ddot{u}\subb+(\cb+\rr\subs\transpose\Cm\subs \rr\subs)\dot u\subb  + \rr\subs\transpose\Km\subs \rr\subs u\subb+f\subb=-m\subb\ddot u\subg +\rr\subs\transpose\Cm\subs\dot{\uu}\subs +\rr\subs\transpose\Km\subs \uu\subs
\end{gather}
where $u\subb$ is the isolation layer drift; $m\subb$ is the mass of the base; 
$c\subb$ is the isolation layer linear damping coefficient; $f\subb$ is the isolation layer restoring force, which can be modeled in a number of ways, including a Bouc-Wen model \cite{wen1976method}, a bilinear model and various ``equivalent'' linear models (see \fref{fig:hyst}). The ``true'' isolator model here, used to generate the data $\D$, is a Bouc-Wen model with 
\begin{equation}\label{eq:bouc}
\begin{split}
f\subb &=\kpost \ub + (1 - \rk) \Qy W z\\
\dot z &= A \ubd - \beta  \ubd|z|^{\npow}-\gamma z|\ubd||z|^{\npow-1}\\
\end{split}
\end{equation}
where $\kpre$ and $\kpost$ are the isolator pre-yield and post-yield stiffnesses, respectively; $\Qy$ is the isolator yield force; $(1-\rk)\Qy W$, the peak of the non-elastic force, depends on the hardness ratio $\rk = \kpost/\kpre$ and the building weight $W=(m\subscript{s}+m\subb)g$, where $m\subscript{s}$ is the total mass of the superstructure; $z$ is an evolutionary variable; and the Bouc-Wen hysteresis shaping constants are chosen $A=2\beta=2\gamma=\kpre/(\Qy W)$ so that $z$ remains in $[-1,1]$ and produces in $f\subb$ identical loading and unloading stiffnesses.
The isolation-layer linear damping coefficient herein is fixed at $\cb=40\units{kN$\cdot$s/m}$.


\begin{table}[htb!]
	\vspace{0pt}
	\centering
	\caption{Exact and prior distribution of parameters for the 11-story base isolated building. } \label{tab:ex3_prior}
	\vspace{0pt}
	\begin{threeparttable}
		
		\begin{tabular}{@{} l | c | c | l| l|l|l @{}} 
			\hline
			\HeaderTstrut
			\multirow{2}{*}{\textbf{Parameter}} & \textbf{True} & \textbf{Prior} & \textbf{Lower }& \textbf{Upper} & \multirow{2}{*}{\textbf{Mean}} & \multirow{2}{*}{\textbf{Std.~dev.}}  \\
			\HeaderBstrut & \textbf{value} & \textbf{Distribution} & \textbf{bound} & \textbf{bound} & & \\
			\hline\Tstrut
			$\kpost$ (kN/m) & 750\hphantom{.0000} & Trunc.~Gaussian& 700\hphantom{.00} & 800\hphantom{.00} & 780\hphantom{.00} & 20\hphantom{.0000} \\
			$\Qy$ (\%) & \hphantom{00}5\hphantom{.0000} & Uniform & \hphantom{00}4.5\hphantom{0} & \hphantom{00}6.5\hphantom{0} & \hphantom{00}5.5\hphantom{0} & \hphantom{0}0.5774\\ 
			\Bstrut	$\rk$ & \hphantom{00}0.1667 & Uniform & \hphantom{00}0.16 &  \hphantom{00}0.20 & \hphantom{00}0.18 & \hphantom{0}0.0115\\
			\hline
		\end{tabular}
	\end{threeparttable}
\end{table}


The uncertain parameters $\thetaa=\{\kpost,\Qy,\rk\}\transpose$ are the post-yield stiffness, the isolator yield force as a percentage off building weight, and the hardening ratio.  The assumed priors for these uncertain parameters are given in Table~\ref{tab:ex3_prior}, along with their ``true values'' \cite{skinner1993} --- which are used to generate the measurement data --- that 
give the isolated structure first mode a typical large-displacement isolation period of $T_1^\mathrm{i} =2.76\units{s}$, and a linear viscous damping ratio of 5.5\% (not including any energy dissipation from the isolator hysteresis).
The ground excitation is the 1940 El Centro earthquake record (peak ground acceleration 3.43 m/s) of 30\units{s} duration with a sampling rate of 20 Hz recorded at Imperial Valley Irrigation District substation in El Centro, California.
The absolute (horizontal) acceleration of the roof, specified by DOF $u_{97}$ in \fref{fig:building}, with a sampling rate of 20\units{Hz} is used as the model output, to which is added a zero-mean Gaussian measurement-noise pulse process with a standard deviation that is 20\% of the standard deviation of the exact response.

The MultiNest algorithm is applied using 500 samples. The LLA-IS algorithm is implemented with a Gaussian importance density formed at each iteration using the mean of the current samples and twice their standard deviations, $\gamma_\mathrm{S}=0.5N$, and an initial sample size of $N=100$. The likelihood levels are decided at every iteration based on the fraction $f_i^\mathrm{reject}=\min(0.5,0.075+0.025i)$. The LLA-SS algorithm is implemented with $\Theta_k^{(s_k)}=\left\{F^{-1}_{\theta_k}\left(\theta_k\right)|\theta_k \in \Big(\frac{s_k-1}{\nstk{k}},\frac{s_k}{\nstk{k}}\Big]\right\}$, where $\nstk{k}=5$ so $\nst=5^3=125$, and $F^{-1}_{\theta_k}$ is the inverse probability distribution of $\theta_k$, $k=1,\dots,3$. For this algorithm, the likelihood levels are decided at every iteration based on the fraction $f_i^\mathrm{reject}=0.25$. 
The LLA-MCMC algorithm is implemented with an initial sample size of 100 and, at each iteration, the 10 samples with the lowest likelihoods
are rejected and 10 new samples with higher likelihoods are added to the sample pool using the Metropolis-Hastings algorithm. The proposal density in the Metropolis-Hastings algorithm is assumed Gaussian with standard deviations of 5\units{kN/m}, 0.5\%, and 0.005 for the parameters $\kpost$, $\Qy$, and $\rk$, respectively.
The stopping criteria are used as in Example I except that the total number of function evaluations is limited to 7000 and $\Delta \Evd_\mathrm{tol}=0.1\%$ to limit computational cost. 

Since the exact marginal likelihood is not available for this example, the difference relative to a Monte Carlo estimation is provided for the proposed methods, nested sampling, and MultiNest in Table~\ref{tab:ex3_evd}. 
The relative difference in the marginal likelihood or evidence is calculated using
\begin{equation}
\text{Relative difference }=\Bigg\lvert\frac{\log \widehat\Evd\superscript{LLA}-\log \widehat\Evd\superscript{MC}}{\log \widehat\Evd\superscript{MC}}\Bigg\rvert
\end{equation}
where $\widehat\Evd\superscript{MC}$ is estimated with Monte Carlo sampling using many likelihood evaluations. Figure~\ref{fig:evdMC_LLA_ExII} is used to determine the sufficient number of likelihood evaluations needed to produce a stable estimate, which is then compared to evidence estimates from other methods.  (In this and subsequent examples, the number of function evaluations are reported as an order-of-magnitude to be consistent with Example I, since they can vary slightly depending on the random seed provided.)
\begin{figure}[htb!]
	\centerline{
		\includegraphics[scale=0.53]{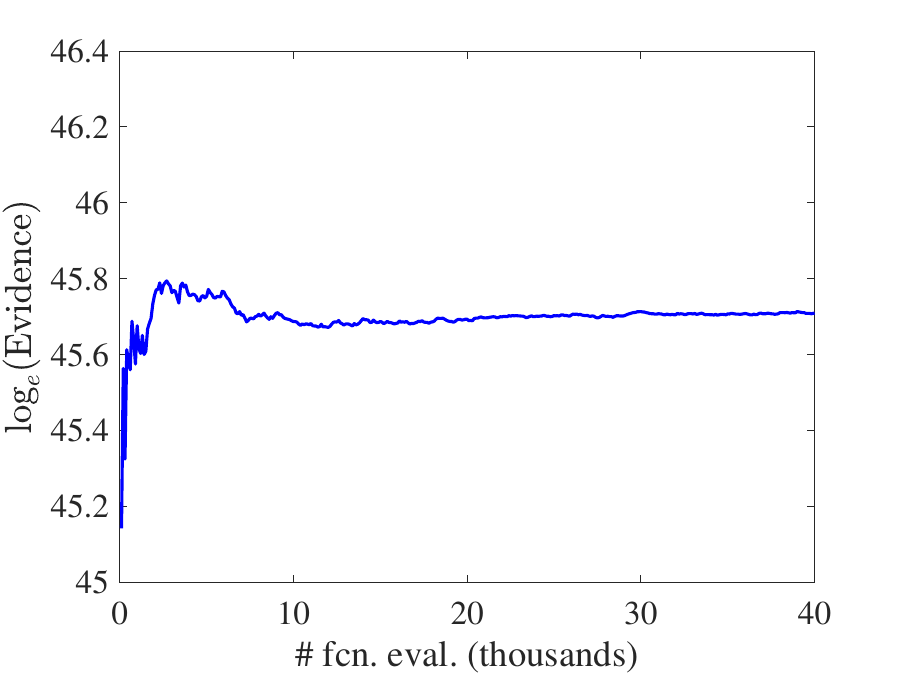}
	} \caption{Estimation of evidence using Monte Carlo approach with increasing number of likelihood evaluations.
    } \label{fig:evdMC_LLA_ExII}\vspace{-5pt}
\end{figure}
Among the three algorithms, the performance of LLA-IS is affected by the assumption of importance density at each iteration since the importance density is not efficient in drawing samples from the high-likelihood region. 
and performs poorly compared to all other algorithms. This shows that Gaussian importance densities might not be appropriate for all problems and a judicious choice of $\pdfis(\thetaa)$ is required depending on the problem. While the MultiNest algorithm performs well, with a relative error of 0.28\%, the LLA-SS performs better because of the low-dimensional parameter space and gives a relative difference of only 0.14\% in log evidence from the Monte Carlo estimate that uses about 8 times as many likelihood evaluations. Similarly, the LLA-MCMC algorithm performs better than MultiNest, producing a relative difference of 0.11\% in log evidence compared to a Monte Carlo sampling requiring about nine times as many likelihood evaluations.
Notably, both nested sampling and MultiNest require significantly more (30--45\% more) function evaluations than, but without achieving accuracy comparable to, the proposed LLA-SS and LLA-MCMC methods.
The estimated means and standard deviations of the parameters using the LLA-MCMC algorithm are shown in Table~\ref{tab:ex3_param}. 

\begin{table}[htb!]
	\vspace{0pt}
	\centering
	\caption{Comparison of evidence estimates using different algorithms for the second example, compared to conventional Monte Carlo and nested sampling.
    } \label{tab:ex3_evd}
	\vspace{0pt}
	\begin{tabular}{@{} l | c| c | c @{}} 
		\hline
		\HeaderTstrut\HeaderBstrut
		\textbf{Method} & \textbf{\# fcn.~eval.} & \textbf{log\textit{\textsubscript{e}}(Evidence)} & \textbf{Relative difference (\%)} \\
		\hline\Tstrut
		Monte Carlo & \hphantom{$\sim$}40,000 & 45.6407 & --- \\
		Nested sampling & $\sim$ 8,000 & 45.1889 & 0.99\\
		MultiNest & $\sim$\hphantom{0}8,000 &  45.7081 & 0.15 \\
		LLA-IS & $\sim$\hphantom{0}6,500 & 43.4989 & 4.69\\
		LLA-SS & $\sim$\hphantom{0}4,900 &  45.7038 & 0.14 \\
		\Bstrut LLA-MCMC & $\sim$\hphantom{0}4,400 & 45.5899 & 0.11 \\
        \hline
	\end{tabular}
\end{table}

\begin{table}[htb!]
	\vspace{0pt}
	\centering
	\caption{Posterior means and standard deviations of the parameters for the 11-story base isolated building, computed using the LLA-MCMC algorithm.} \label{tab:ex3_param}
	\vspace{0pt}
	
	\begin{tabular}{@{} l | c | c | c @{}} 
		\hline
        \HeaderTstrut
		\multirow{2}{*}{\textbf{Parameter}} & \textbf{True} & \textbf{Posterior} & \textbf{Posterior}  \\
		\HeaderBstrut & \textbf{value} & \textbf{Mean} & \textbf{Std.~dev.} \\
		\hline\Tstrut
		$\kpost$ (kN/m)  & 750\hphantom{.0000} & 766.84\hphantom{00} & 17.12\hphantom{00} \\
		$\Qy$ (\%)  & \hphantom{00}5\hphantom{.0000} & \hphantom{00}5.04\hphantom{00}& \hphantom{0}0.25\hphantom{00} \\ 
		\Bstrut $\rk$   &  \hphantom{00}0.1667 & \hphantom{00}0.1677 & \hphantom{0}0.0036\\
		\hline
	\end{tabular}
\end{table}

Next, a Bayesian model selection exercise is performed for this example to determine the restoring force model that best reproduces the measured data. The candidate models are Bouc-Wen, bilinear, and a linear approximation of the bilinear model according to the AASHTO 
guidelines~\cite{subcommittee2010guide}. 
\begin{figure}[htb!]
	\begin{center}
		\includegraphics[scale=0.3]{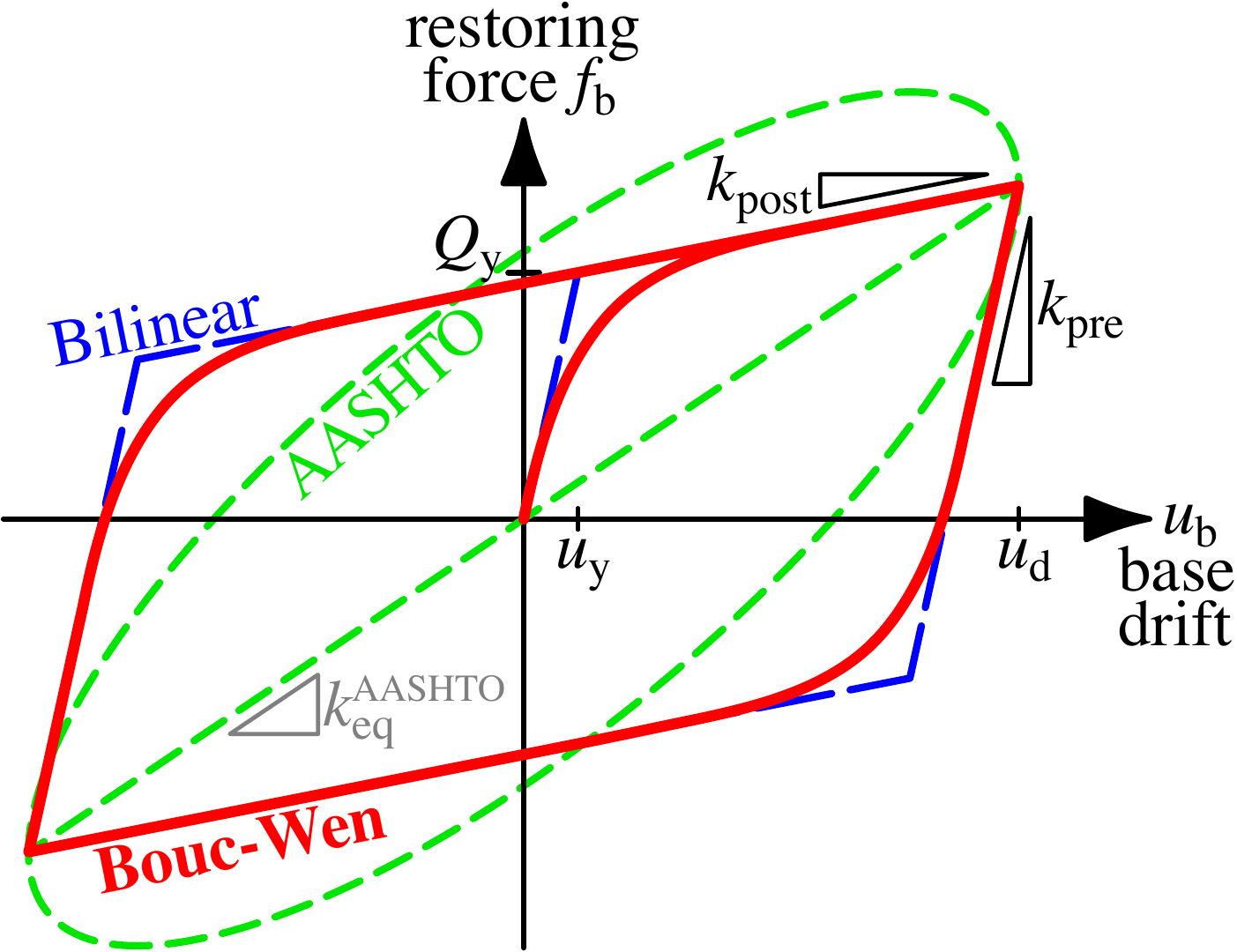}
	\end{center}
	\caption{Models for isolator hysteresis. 
    } \label{fig:hyst}
\end{figure}
The bilinear model can be approximated by a Bouc-Wen model \eqref{eq:bouc}
as $\npow\rightarrow\infty$; this study uses $\npow=100$, which is sufficiently large to produce a bilinear hysteresis loop. The AASHTO model approximates the energy dissipation in each cycle of the bilinear model, approximating the isolator restoring force with
\begin{equation}
\begin{split}
f\subscript{b}&=c\subscript{eq}\dot{u}\subb+\keq\ub\\
&=2\zetaeq\sqrt{\keq (m\subscript{b}\mbox{}+m\subscript{s})}\dot{u}\subb+\keq\ub.\\
\end{split}
\end{equation}
The equivalent damping ratio $\zetaeq$ and equivalent stiffness $\keq$ are given by
\begin{equation}
\begin{split}
&\zetaeq=\frac{2(1-\rk)(1-\rd^{-1})}{\pi[1+\rk(\rd-1)]}\\
&\keq=\frac{\kpre}{\rd}{[1+\rk(\rd-1)]}
\end{split}
\end{equation}
where $\rd=u\subscript{d}/u\subscript{y}$ is the shear ductility ratio of the design displacement $u\subscript{d}$ to the yield displacement $u\subscript{y}$.

The marginal likelihoods are calculated using the LLA-MCMC algorithm. With equal prior model probabilities, the posterior model probabilities are calculated using Bayes' theorem \eqref{eq:mod_sel_defn} and shown in Table~\ref{tab:ex3post}. The model selection correctly assigns a posterior probability of 1.0 to Bouc-Wen model. Also note that, if the Bouc-Wen model is absent from the candidate model set, the Bayesian model selection chooses the second-best bilinear model as the correct one. This is expected because the bilinear model is the only remaining nonlinear model.
\begin{table}[htb!]
	\vspace{0pt}
	\centering
	\caption{Posterior model probabilities for the hysteretic isolation layer in the 11-DOF base-isolated building, estimated using the LLA-MCMC algorithm.} \label{tab:ex3post}
	\vspace{0pt}
	\begin{tabular}{@{} l | c | c @{}} 
		\hline
		\HeaderTstrut\HeaderBstrut
		\textbf{Model} & \textbf{log\textit{\textsubscript{e}}(Evidence)} & \textbf{$\Prob{\Mk|\D}$} \\
		\hline\Tstrut
		Bouc-Wen & \hphantom{$-$0}45.6893 & $\approx$ 1.0\\
		Bilinear & \hphantom{0}$-$46.1149 & $\approx$ 0.0 \\
		\Bstrut		AASHTO &  $-$821.0503 & $\approx$ 0.0 \\
		\hline
	\end{tabular}
\end{table} 

\subsection{Example III: Flow past a cylinder}
In this example, 2D fluid flow in a channel 
past a slightly off-center cylinder 
is considered, where the fluid inflow velocities are uncertain. Simulations are solved using the finite element method implemented in the FEniCS software package \cite{langtangen2016solving}, adapted from one of the examples in its documentation. \fref{fig:flow} shows the cross-section of the channel of width $h = 0.41\units{m}$
and the cylinder of diameter $0.10\units{m}$. 
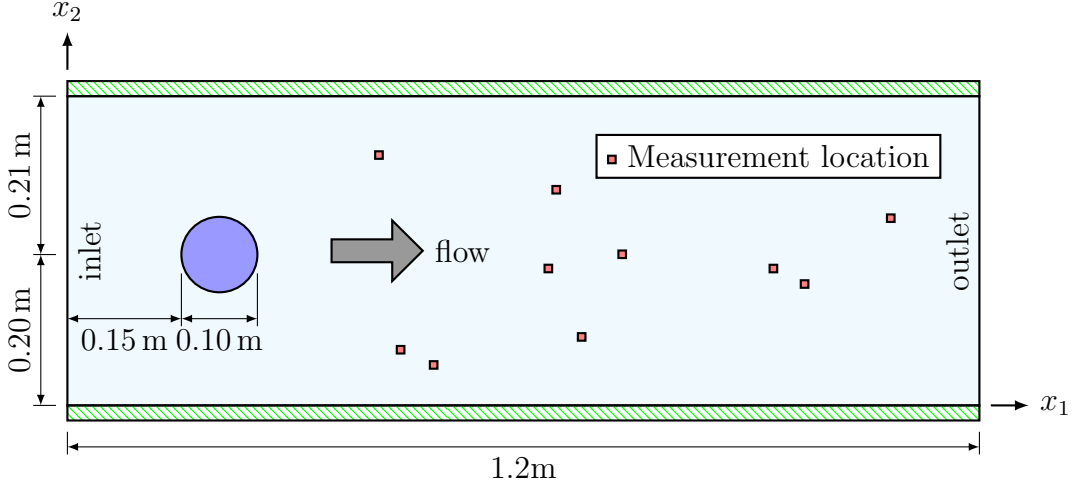
\begin{figure}[htb!]
	\centering
	\begin{tikzpicture}
	
	\def\wallthickness{0.2}
	\def\dimstart{0.15}
	\def\dimend{0.45}
	\def\dimloc{0.35}
	\def\axislen{0.5}
	\begin{scope}[local bounding box=channel]
		\draw[thick,fill=cyan!5] (0,0) rectangle (12,4.1);
	\end{scope}
	\begin{scope}[local bounding box=cylinder]
		\draw[fill=blue!40,thick] (2,2) circle (0.5);
	\end{scope}
	\draw[thick,pattern=north west lines, pattern color=green] (channel.south west) + (0,-\wallthickness) rectangle (channel.south east);
	\draw[thick,pattern=north west lines, pattern color=green] (channel.north west) + (0,\wallthickness) rectangle (channel.north east);

	\draw[-latex,thick] ($ (channel.north west) + (0,\wallthickness+\dimstart) $) -- ++(0,\axislen) node[anchor=south] {$x_2$};
	\draw[-latex,thick] ($ (channel.south east) + (\dimstart,0) $) -- ++(\axislen,0) node[anchor=west] {$x_1$};

	\draw (channel.south west)+(-\dimstart,0) -- ++(-\dimend,0);
	\draw (channel.north west)+(-\dimstart,0) -- ++(-\dimend,0);
	\draw (channel.south west |- cylinder.center)+(-\dimstart,0) -- ++(-\dimend,0);
	\coordinate (ycylinder) at ($ (cylinder.center -| channel.south west) + (-\dimloc,0) $);
	\draw[latex-latex] ($ (channel.south west) - (\dimloc,0) $) -- node[above,rotate=90] {0.20\units{m}} (ycylinder);
	\draw[latex-latex] (ycylinder) -- node[above,rotate=90] {0.21\units{m}} ($ (channel.north west) - (\dimloc,0) $);

	\draw ($ (channel.south west) - (0,\wallthickness) $)+(0,-\dimstart) -- ++(0,-\dimend);
	\draw ($ (channel.south east) - (0,\wallthickness) $)+(0,-\dimstart) -- ++(0,-\dimend);
	\draw[latex-latex] ($ (channel.south west) - (0,\wallthickness+\dimloc) $) -- node[below,rotate=0] {1.2m} ($ (channel.south east) - (0,\wallthickness+\dimloc) $);
	
	\coordinate (ydimcylinder) at ($ (cylinder.south west) + (0,-\dimend) $);
	\draw ($ (cylinder.south west) + 0.25*(cylinder.north) - 0.25*(cylinder.south) $) -- (ydimcylinder);
	\draw ($ (cylinder.south east) + 0.25*(cylinder.north) - 0.25*(cylinder.south) $) -- ($ (cylinder.south east) + (0,-\dimend) $);
	\coordinate (ydimcylinder2) at ($ (ydimcylinder) + (0,\dimend) - (0,\dimloc) $);
	\draw[latex-latex] (channel.west|-ydimcylinder2) -- node[below,rotate=0] {0.15\units{m}} (ydimcylinder2);
	\draw[latex-latex] (ydimcylinder2) -- node[below,rotate=0] {0.10\units{m}} (cylinder.east|-ydimcylinder2);
	
	\node[rotate=90,anchor=north] at (channel.west) {inlet};
	\node[rotate=90,anchor=south] at (channel.east) {outlet};
	
	\node[thick,single arrow,draw=black,fill=black!40,minimum height = 1.2 cm] (arrow) at ($ (channel.west) + (4,0) $) {};
	\node[anchor=west] at (arrow.east) {flow};

	\draw[thick, fill = red!50] (10.8835,2.5321) rectangle (10.7835,2.4321);
	\draw[thick, fill = red!50] (4.8691,0.5865) rectangle (4.7691,0.4865);
	\draw[thick, fill = red!50] (6.8166,0.9579) rectangle (6.7166,0.8579);
	\draw[thick, fill = red!50] (4.4342,0.7880) rectangle (4.3342,0.6880);
	\draw[thick, fill = red!50] (9.3395,1.8647) rectangle (9.2395,1.7647);
	\draw[thick, fill = red!50] (6.4819,2.9094) rectangle (6.3819,2.8094);
	\draw[thick, fill = red!50] (7.3498,2.0538) rectangle (7.2498,1.9538);
	\draw[thick, fill = red!50] (9.7498,1.6593) rectangle (9.6498,1.5593);
	\draw[thick, fill = red!50] (6.3782,1.8652) rectangle (6.2782,1.7652);
	\draw[thick, fill = red!50] (4.1494,3.3694) rectangle (4.0494,3.2694);
	
	\def\egendoffset{0.5}
	\node[draw,rectangle,thick,fill=white,anchor=north east] (legend) at ($ (channel.north east) - (\egendoffset,\egendoffset) $) {%
		$\vcenter{\hbox{%
		\begin{tikzpicture}
		\draw[thick, fill = red!50,inner sep=0,outer sep=0] (0,0) rectangle (0.1,0.1);
		\end{tikzpicture}%
		}}$%
	~Measurement location};
	\end{tikzpicture}
	\caption{ The setup for Example III where flow, with a parabolic velocity profile, enters the domain from left side. The measurement locations are shown using small squares.} \label{fig:flow}
\end{figure}

The Navier-Stokes equation for an incompressible fluid, along with the mass conservation equation, are
\begin{equation}\label{eq:nst}
\begin{split}
\frac{\partial \uu}{\partial t}+(\uu \cdot \nabla)\uu & = -\frac{1}{\rho}\nabla p + \nu \nabla^2\uu\\
\nabla \cdot \uu & = 0\\
\end{split}
\end{equation}
where $\uu(\x;t)=[u_1(\x;t),u_2(\x;t)]\transpose$ is the velocity vector at position $\x=[x_1,x_2]\transpose$; $p$ is the pressure; $\rho$ is the fluid density; $\nu$ is the kinematic viscosity of the fluid; and the gradient operator $\nabla=[\partial/\partial\x]^\mathrm{T}=\left[{\partial}/{\partial x_1}, {\partial}/{\partial x_2}\right]\transpose$. Equations \eqref{eq:nst} are solved using the incremental pressure correction scheme in FEniCS.
The inlet velocity profile is assumed parabolic
\begin{equation}
\uu([0,x_2]\transpose;t,\theta)= \left[\frac{4x_2(h-x_2)}{h^2}\theta, 0\right]\transpose 
\end{equation}
where $\theta$ is the uncertain peak inlet velocity.
The velocity measurements are generated with $\theta_\text{true}=1.2\units{m/s}$, which results in a flow with a Reynolds number of 80. The horizontal velocities are sampled at ten points $\x_1,\dots,\x_{10}$ randomly chosen downstream from the cylinder, starting at 1\,s after the start of the flow and stopping before 1.6\units{s} with an interval of 0.05\units{s}, resulting in a total of $n=120$ velocities $u_{1,l,k}^\text{true}=u_1(\x_l;t_k,\theta_\text{true})$, where $l=1,\dots,10$, $k=1,\dots,12$, and $t_k=0.95\units{s}+(0.05\units{s})k$. These velocities are stacked into a vector $\dd_\text{exact}$ and then corrupted by additive zero-mean Gaussian measurement noise $\textbf{v}$:
    \begin{equation}\label{eq:stackedvelocities}
    \dd = \dd_\text{exact} + \mathbf{v}, \qquad
    \dd_\text{exact}=[u_{1,1,1}^\text{true},\dots,u_{1,1,12}^\text{true},~~\dots~~,u_{1,10,1}^\text{true},\dots,u_{1,10,12}^\text{true}]\transpose,
    \end{equation}
where the elements of $\mathbf{v}$ are independent and identically distributed zero-mean Gaussian variables, each with a standard deviation that is 20\% of the standard deviation of the 120 measurements --- \ie 20\% of 
$\sqrt{\mathrm{Var}(\dd_{\mathrm{exact}})} = \{\tfrac{1}{n-1}[\sum_{k=1}^{n}d_{\mathrm{exact},k}^2] - \tfrac{1}{n(n-1)}[\sum_{k=1}^{n}d_{\mathrm{exact},k}]^2\}^{1/2}$, where $d_{\mathrm{exact},k}$ is the $k$\textsuperscript{th} element of $\dd_{\mathrm{exact}}$.
\begin{figure}[htb!]
	\centering
	\begin{subfigure}[t]{0.95\textwidth}
		\centering
		\includegraphics[scale=0.4]{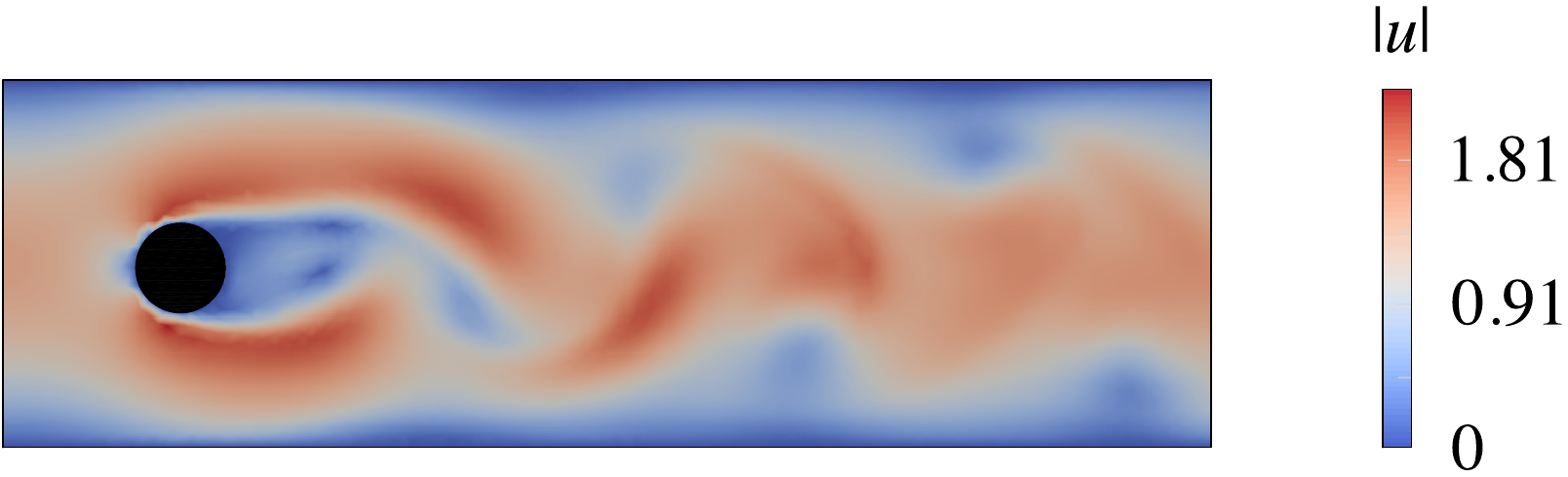}
		\caption{Horizontal velocity magnitude distribution (m/s).}
	\end{subfigure}%

	\begin{subfigure}[t]{0.95\textwidth}
		\centering
		\includegraphics[scale=0.4]{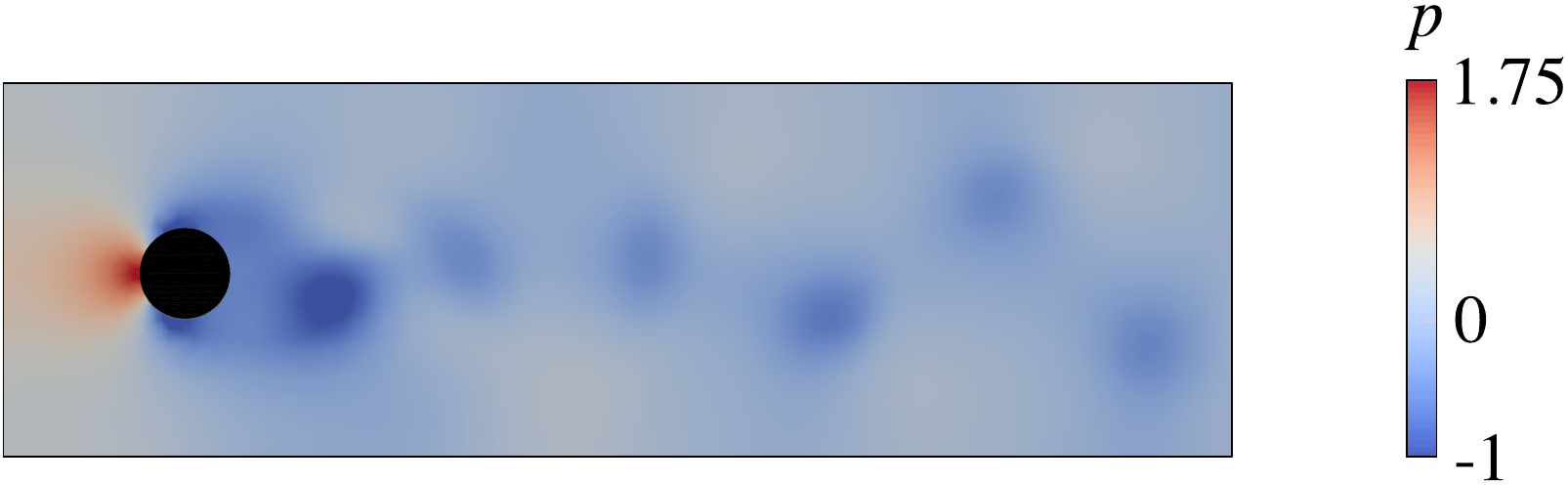}
		\caption{Pressure distribution (N/m$^2$).}
	\end{subfigure}
	\caption{Typical velocity and pressure distributions in Example III at $t=1.5\units{s}$.}
\end{figure}

For prediction purposes, the measured velocity is assumed to contain both an additive measurement noise and a multiplicative modeling error. 
Herein, the multiplicative error is used as in Cheung \textit{et al.}~\cite{cheung2011bayesian}, Oliver and Moser~\cite{oliver2011bayesian}, and Edeling \textit{et al.}~\cite{edeling2014bayesian,edeling2014predictive}. In this uncertainty model, the measured velocity is assumed to be given by 
\begin{equation}
u_1^\mathrm{meas}(\x;t,\theta)=E_\mathrm{m}(\x)u_1(\x;t,\theta)+e
\end{equation}
where $u_1(\x;t,\theta)$ is the velocity obtained by solving \eqref{eq:nst}; the multiplicative modeling error given by $E_\mathrm{m}(\x)$ is used to generate the modeling uncertainty; and $e$ is the measurement error. 
The required covariances are then written as
\begin{equation}\label{eq:cov2}
\begin{split}
k_{{uu}}(\x,\x^\prime;t|\theta)&=u_1(\x;t,\theta)k_{\!_{E_\mathrm{m}}}(\x,\x^\prime;t|\theta)u_1(\x^\prime;t,\theta)\\
\end{split}
\end{equation}
where the measurement error is assumed independent. 
The multiplicative error $E_\mathrm{m}$ is assumed herein to be a time-invariant, spatial Gaussian process with unit mean and covariance function given by
\begin{equation}\label{eq:cov3}
k_{\!_{E_\mathrm{m}}}(\x,\x^{\prime};t|\thetaa_\mathrm{hyp})=\sigma^2\exp\left[-\sum_{i=1}^2\left(\frac{x_i-x_i^\prime}{l_i}\right)^2\right]
\end{equation}
where the hyperparameters are $\thetaa_\mathrm{hyp}=[\sigma, l_1, l_2]\transpose$. 
	The measurement errors in $\dd$ are assumed distributed as $\mathcal{N}(\mathbf{0},\Sigm_e)$, where $\Sigm_e=\sigma_e \mathbf{I}$ is a diagonal matrix.
A Gaussian likelihood function is constructed, next, as 
\begin{equation}
\pdf(\D|\theta)=\frac{1}{(2\pi)^{n/2}|\Sigm|^{1/2}}\exp\left\{-\frac{1}{2}\left[\uu_1(\theta)-\dd\right]\transpose\Sigm^{-1}\left[\uu_1(\theta)-\dd\right]\right\}
\end{equation}
where $\uu_1(\theta)=[u_{1,1,1}^\theta,\dots,u_{1,1,12}^\theta,~~\dots~~,u_{1,10,1}^\theta,\dots,u_{1,10,12}^\theta]\transpose$ 
consists of the $n=120$ velocities predicted using parameter $\theta$ corresponding to each element in $\dd$ where $u_{1,l,k}^\theta \equiv u_1(\x_l;t_k,\theta)$; $\Sigm=\Sigm_e+\Sigm_{E_\mathrm{m}}$ is the covariance matrix of the likelihood, where $\Sigm_{E_m}$ is formed using the covariance of the Gaussian process $E_\mathrm{m}$ in \eqref{eq:cov3}, \ie where the $(p+12[j-1],q+12[k-1])$ element of $\Sigm_{E_m}$ is 
\begin{equation}
\left.
\sigma^2u_1(\x_j;t_p)u_1(\x_k;t_q)
\middle/
\prod_{i=1}^2\exp\left(\frac{x_{j,i}-x_{k,i}}{l_i}\right)^2
\right.
\end{equation}
for $j,k=1,\dots,10$ and $p,q=1,\dots,12$. The hyperparameters are assumed here as $\sigma=0.2\sqrt{\mathrm{Var}(\dd)}$, and $l_1=l_2=1\units{m}$ to model the correlated and corrupted measurement data. 
The prior for peak inlet velocity $\theta$ is assumed to be a truncated Gaussian with mean 1.25\units{m/s} and standard deviation 0.5\units{m/s}, truncated at 1.0 and 1.5\units{m/s}.

The LLA-IS algorithm is implemented with a truncated Gaussian importance density formed at each iteration, using the mean of the current samples and a standard deviation of 0.25\units{m/s}, truncated at 1.0 and 1.5\units{m/s}, and with $\gamma_\mathrm{S}=0.5N$ and an initial sample size of $N=25$. The likelihood levels are decided at iteration $i$ based on the fraction $f_i^\mathrm{reject}=\min(0.5,0.075+0.025i)$. The LLA-SS algorithm is implemented again with $\Theta^{(s_1)}=\left\{F^{-1}_{\theta}\left(\theta\right)|\theta \in \Big(\frac{s_1-1}{\nstk{1}},\frac{s_1}{\nstk{1}}\Big]\right\}$, where $\nst=\nstk{1}=5$, and $F_{\theta}$ is the probability distribution of the peak inlet velocity $\theta$; for this algorithm, the likelihood levels are decided at every iteration based on the fraction $f_i^\mathrm{reject}=0.25$.
The LLA-MCMC algorithm is implemented with an initial sample size of 100 and, at each iteration, the 20 samples with the lowest likelihoods are rejected. 20 new samples with likelihoods larger than the highest likelihood value of the rejected samples are added to the sample pool using the Metropolis-Hastings algorithm with a proposal density that is a truncated Gaussian with a standard deviation 0.1\units{m/s} and truncated at 1\units{m/s} and 1.5\units{m/s}. The stopping criteria are the same as in Example I except $\Delta\Evd_\mathrm{tol}=1\%$ and a maximum of 5000 likelihood evaluations.

The accuracies of the proposed algorithms, again compared to Monte Carlo and nested sampling, are shown in Table~\ref{tab:ex2_evd}. 
The sufficient number of likelihood function evaluations for the Monte Carlo approach is estimated similar to Example II. 
All three algorithms give small relative difference in the log evidence here as shown in Table~\ref{tab:ex2_evd}.
The LLA-SS performs well, as the dimension of the parameter space is low. The LLA-IS and LLA-MCMC algorithms in this example also produce small relative differences compared to a Monte Carlo sampling with 5.7 times as many likelihood evaluations. When compared to nested sampling with similar numbers of likelihood evaluations, the proposed algorithms --- notably LLA-MCMC --- show improved accuracy of almost two orders of magnitude. 

\begin{table}[htb!]
	\vspace{0pt}
	\centering
	\caption{Comparison of evidence estimates, using different algorithms for the third example, compared to conventional Monte Carlo and nested sampling.
    } \label{tab:ex2_evd}
	\vspace{0pt}
	\begin{tabular}{@{} l | c| c| c @{}} 
		\hline
		\HeaderTstrut\HeaderBstrut
		\textbf{Method} & \textbf{\# fcn.~eval} & \textbf{log\textit{\textsubscript{e}}(Evidence)} & \textbf{Relative difference (\%)} \\
		\hline\Tstrut
		Monte Carlo & \hphantom{$\sim$}20,000 & 252.5201& --- \\
		Nested sampling & $\sim$\hphantom{0}4,000& 252.2854 & 0.0929 \\
		LLA-IS & $\sim$\hphantom{0}3,500 &252.5498 & 0.0118\\
		LLA-SS & $\sim$\hphantom{0}3,500 & 252.5067  & 0.0053\\
		\Bstrut LLA-MCMC & $\sim$\hphantom{0}3,500 & 252.5249 & 0.0019 \\

        
        

        
        
        
        
        
        
        
        
		\hline
	\end{tabular}
\end{table}


\subsection{Example IV: Steady state heat conduction within an inhomogeneous plate}
This example explores the LLA-MCMC algorithm on a high-dimensional problem, for which it is more suited than the other two algorithms that may suffer significantly from the curse of dimensionality, as discussed in Section~\ref{sec:method}%
, because finding a good LLA-IS importance density in high dimensions is difficult and, as implemented, the number of strata in LLA-SS grows exponentially with dimension.
Although a good rejection algorithm like MultiNest is efficient in low dimensions, an exponential scaling emerges \cite{Handley2015} because the parameter space volume expands exponentially with the number of dimensions, and the number of samples must be increased to obtain accurate estimates \cite{Feroz2013}.

A square thin plate of dimension 2\,m$\times$2\,m is composed of an inhomogeneous material with a coordinate system origin at its center. One edge (at $y=-1\units{m}$) is subjected to a fixed temperature of 100$^\circ$C; the other three edges are thermally insulated.
A point source of $Q=25\units{W}$ 
is located in the interior at $x=-0.5\units{m}$, $y=0$. The material thermal conductivity $k(x_1,x_2)=k_0 \exp[z(x_1,x_2)]$ 
is assumed uncertain, where $z(x_1,x_2)$ is an underlying unitless zero-mean Gaussian field and, herein, $k_0 = 1\units{W/(m\textsuperscript{2}$\cdot$K)}$. To keep the problem simple, the covariance kernel of the latent field is assumed as
\begin{equation}\label{eq:correl}
\Exp[z(x_1,x_2)z(y_1,y_2)]=\sigma^2\exp\left(\sum_{k=1}^2-\frac{|x_k-y_k|}{l_k}\right)
\end{equation}
where $l_1$ and $l_2$ are the correlation lengths; $\sigma=1.5$ is used herein. 
Next, the latent random field $z(x_1,x_2)$ is expressed using a Karhunen-Lo\'eve expansion truncated at the $M$\textsuperscript{th} term:
\begin{equation}\label{eq:kl}
z(x_1,x_2)=\sum_{k=1}^M  \omega_k \xi_k f_k(x_1,x_2)
\end{equation}
where the $\omega_k^2$ 
and $f_k(x_1,x_2)$ are the eigenvalues and eigenfunctions, respectively, of the covariance kernel; the $\xi_k$ are standard unit normal variables; $M=100$ is selected to capture 99.6\% of the total energy. Hence, the dimension of the uncertain parameter vector is 100 in this example. 
To generate the temperature measurements in the steady-state condition, correlation lengths $l_1=l_2=0.75\units{m}$ are used; 25 random points are selected on the plate, at which temperatures are measured and corrupted by 10\% Gaussian additive sensor noise (\ie the noises at the sensors are independent and identically distributed zero-mean Gaussian random variables with a standard deviation that is 10\% of the standard deviation of the noise-free measurements. 
A typical temperature distribution is shown in \fref{fig:temp_dist}. 

Using a Gaussian likelihood and (incorrectly) guessing the correlation lengths both to be 1.0\units{m}, the marginal likelihood is estimated with an initial sample size of 100 using the LLA-MCMC algorithm. The stopping criterion is solely that the change in evidence be smaller than 0.1\%. The result is compared in Table~\ref{tab:ex4_evd} to Monte Carlo sampling with $6\times10^5$ likelihood evaluations and to a nested sampling algorithm that is run to the same level of $\widehatchisub{\mathrm{final}}^\mathrm{MCMC}$ as LLA-MCMC achieves. The sufficient number of likelihood function evaluations for the Monte Carlo approach is estimated similar to Example II. LLA-MCMC again provides better accuracy compared to nested sampling using only two-thirds as many likelihood evaluations.
\begin{figure}[htb!]
	\begin{center}
	\begin{tikzpicture}
	\node[inner sep=0pt] (pic) at (0,0)
	    {\includegraphics[scale=0.5]{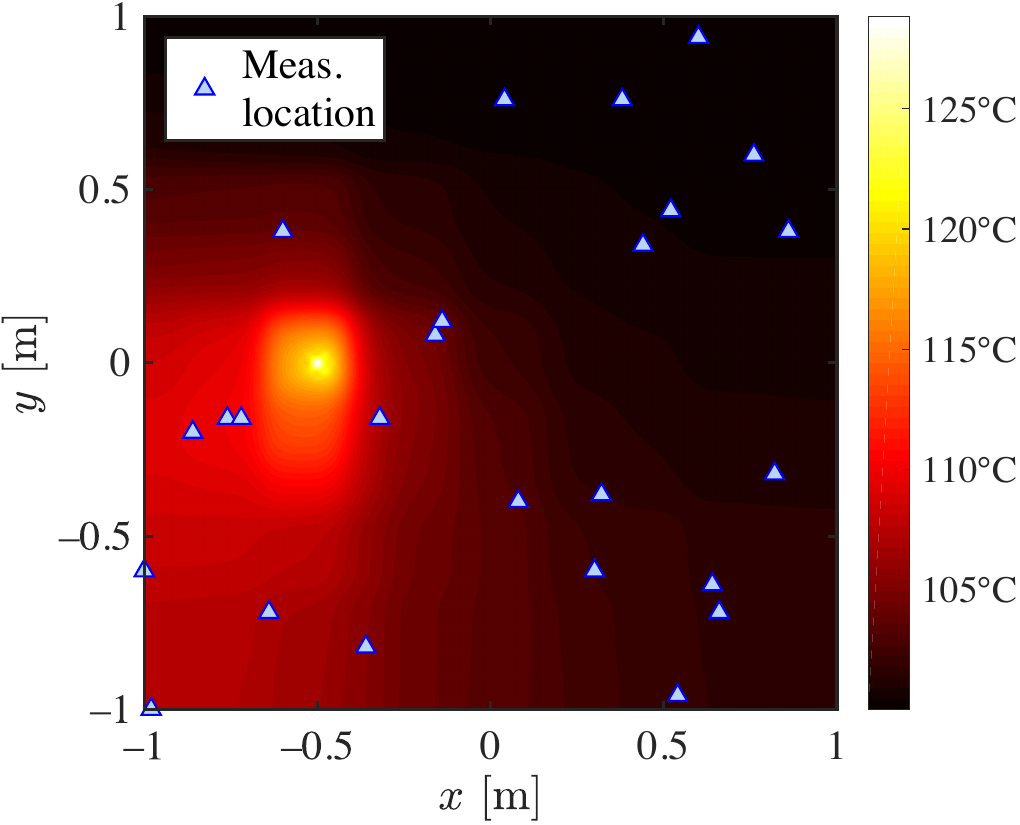}};
	\draw[pattern=north east lines, pattern color=black, draw = none,thick] (-3.1,3.25) rectangle (2.8,3.75);
	\node[draw=none] at (0,4)   () {Fixed $T=100^\circ$ C};
	\end{tikzpicture}
	\end{center}
	\caption{Typical temperature distribution with 25 random measurement locations shown.} \label{fig:temp_dist}
\end{figure}
\begin{table}[htb!]
	\vspace{0pt}
	\centering
	\caption{Comparison of evidence estimate using LLA-MCMC algorithm for Example IV, compared to conventional Monte Carlo and nested sampling.
    } \label{tab:ex4_evd}
	\vspace{0pt}
	\begin{tabular}{@{} l | c| c | c @{}} 
		\hline
		\HeaderTstrut\HeaderBstrut
		\textbf{Method} & \textbf{\# fcn.~eval.} & \textbf{log\textit{\textsubscript{e}}(Evidence)} & \textbf{Relative difference (\%)} \\
		\hline\Tstrut
		Monte Carlo & \hphantom{$\sim$} 60$\times$10\textsuperscript{4} & $-$30.2160 & --- \\
		Nested sampling & $\sim$ 12$\times$10\textsuperscript{4} & $-$31.1172 & 2.98\\
		\Bstrut LLA-MCMC & $\sim$ \hphantom{0}8$\times$10\textsuperscript{4} & $-$30.1295 & 0.29 \\
		\hline
	\end{tabular}
\end{table}

A model selection problem is also implemented for this example, where each model is a polynomial chaos expansion (PCE) \cite{ghanem1991stochastic} of different orders. 
(To allow for easier modeling using lower-order PCE --- as the intent is to evaluate the methods proposed herein, not PCE --- the number of terms in the Karhunen-Lo\'eve expansion for model selection is reduced to $M=22$, which captures 92\% of the total energy.) 
Bayesian model selection is performed to evaluate the suitability of second-order, third-order, and sparse fourth-order PCEs with Hermite polynomials. For the second- and third-order PCEs, an ordinary least-squares minimization is used to estimate the coefficients. A least angle regression is used to estimate the coefficients of the sparse fourth-order PCE \cite{efron2004least}. 
The result, reported in Table~\ref{tab:ex4_result}, shows that a third-order PCE is more likely than a second-order PCE in this example. However, a fourth-order sparse PCE models the data much more poorly compared to either second- or third-order expansions.




\begin{table}[htb!]
	\vspace{0pt}
	\centering
	\caption{Posterior model probabilities for different orders of polynomial chaos expansion.}  \label{tab:ex4_result}
	\vspace{0pt}
	\begin{tabular}{@{} l | c | c @{}} 
		\hline
		\HeaderTstrut\HeaderBstrut
		\textbf{PCE order} & \textbf{log\textit{\textsubscript{e}}(Evidence)} & \textbf{$\Prob{\Mk|\D}$} \\
		\hline\Tstrut
		second & $-36.1309$ & $0.31$\\
		third & $-35.3233$ & $0.69$\\
		\Bstrut		fourth (Sparse) &  $-51.6090$ & $0.00$ \\
		\hline
	\end{tabular}
\end{table}

\section{Conclusions}\label{sec:conc}
A likelihood level adapted approach has been proposed herein to evaluate the multidimensional integral known as the marginal likelihood, or evidence, that is required in Bayesian model selection. Three algorithms using importance sampling, stratified sampling, and Markov chain Monte Carlo are used to implement this approach, and are compared with Monte Carlo sampling and two nested sampling algorithms (original and MultiNest). 
In the first proposed algorithm (LLA-IS), importance sampling densities are formed adaptively at each iteration of the algorithm to generate samples with increased levels of likelihood; the performance of this algorithm depends, however, on the effectiveness of the assumed importance sampling density and, as a result, LLA-IS can become inefficient in sampling from high-likelihood regions.
The second algorithm (LLA-SS), which uses stratified sampling and focuses on generating more samples from the strata with higher likelihoods, is the most accurate method for models with complicated behavior in low dimension, and can be implemented in parallel for each stratum; 
however, this algorithm can suffer from the curse of dimensionality for problems with high-dimensional parameter spaces. The final algorithm (LLA-MCMC) uses Markov chain Monte Carlo algorithms to sample more from the high likelihood region; for high-dimensional parameter spaces, MCMC algorithms with high acceptance rates, \eg the modified Metropolis-Hastings algorithm \cite{au2001estimation}, can be used to outperform the other approaches. 

The performance and accuracy of these proposed algorithms are compared in four numerical examples.
The first example shows the accuracy of the proposed algorithms for a problem with a known marginal likelihood.
The second example uses an 11-story base-isolated building model, with uncertainties in the nonlinear hysteretic isolation layer, subjected to a historical earthquake; the uncertain isolation-layer parameters provide a multidimensional parameter space.
The third example evaluates the marginal likelihood (evidence) for a model where the underlying physics requires many solution of the Navier-Stokes equations; the proposed algorithms are shown to give results that are within the desired accuracy level relative to those estimated using the well-known nested sampling algorithm.
The fourth example involves heat conduction in a plate with uncertain thermal conductivity, which is modeled as a random field with 100 stochastic dimensions; the LLA-MCMC algorithm applied to this problem shows the superior accuracy of the proposed algorithm when multi-dimensional uncertainty is present.
These examples show the efficacy of the proposed algorithms in different settings and their application to Bayesian model selection. The LLA-IS algorithm, however, needs a better selection strategy to construct the importance sampling densities. Similarly, the application of sparse quadrature to the LLA-SS algorithm needs to be investigated. 
In future work, the proposed algorithms will also be implemented for multi-physics problems, and their efficient parallel implementation will be explored. 





\def\appendixname{Appendix}
\appendix
%
%

	\section{The stochastic error $\errorcomponent_{\mathrm{s}}$}\label{sec:app1}
	The stochastic error is defined as
	\begin{equation}
	\sum_{i=0}^{i_{\mathrm{final}}}\lambda_i\left[\Delta\chi(\lambda_i)-\Delta\widehat{\chii}\right] = \sum_{i=0}^{i_{\mathrm{final}}}\lambda_i \errorcomponent_{\mathrm{s},i}
	\end{equation}
	where each component of the error is $\errorcomponent_{\mathrm{s},i}=\Delta\chi(\lambda_i)-\Delta\widehat\chii$.
\begin{prop}
	For each $i$, the stochastic error component, 
    $\Delta\widehat{\chii}$ converges to $\Delta\chi(\lambda_i)$, and its coefficient of variation (COV) is $\mathcal{O}(N^{-1/2})$ when the samples are uncorrelated and $N$, the number of samples, is large.
\end{prop}
\begin{proof}
	For a large number $N$ of samples, using the strong law of large numbers, $\Delta\widehat{\chii}$ converges to $\Delta\chi(\lambda_i)$ almost surely \cite{feller2008introductionv1} as the mean of a sequence of random variables.
    Hence, $\errorcomponent_{\mathrm{s},i}=\lambda_i[\Delta\chi(\lambda_i)-\Delta\widehat{\chii}]$ is zero-mean for large $N$.
	Further, denote $\mathbb{I}_{i,k}(\thetaa_j)=\mathbb{I}_{\widetilde\Thetaa_i}(\thetaa_j)\left[1-\mathbb{I}_{\widetilde\Thetaa_k}(\thetaa_j)\right]$, \ie an indicator that sample $\thetaa_j$ is in $\widetilde\Thetaa_i$ but outside $\widetilde\Thetaa_k$. Hence,
	\begin{equation}
	\Delta\widehat{\chii}
    =\frac{1}{N}\sum_{j=1}^{N} \mathbb{I}_{i-1,i}(\thetaa_j)
	\end{equation}
	and
	\begin{equation}\label{eq:var_chi}
	\begin{split}
	\Exp\big[(\Delta\widehat\chii-\Delta\chi(\lambda_i))^2\big]&=\Exp\left[\left(\frac{1}{N}\sum_{j=1}^{N} \big[\mathbb{I}_{i-1,i}(\thetaa_j)-\Delta\chi(\lambda_i)\big]\right)^2\right]\\
	&=\frac{1}{N^2}\sum_{j=1}^{N}\sum_{l=1}^{N}\Exp\Big[\big( \left[\mathbb{I}_{i-1,i}(\thetaa_j)-\Delta\chi(\lambda_i)\right]\big)\big( \left[\mathbb{I}_{i-1,i}(\thetaa_l)-\Delta\chi(\lambda_i)\right]\big)\Big]
	\end{split}
	\end{equation}
	Note that the $\mathbb{I}_{i-1,i}(\thetaa_j)$ are Bernoulli random variables with success probabilities $\Delta\chi(\lambda_i)$.
	Assume the samples are uncorrelated, \ie 
	\renewcommand{\qedsymbol}{~}
	\begin{equation}\label{eq:uncor}
	\begin{split}
	\Exp\Big[\big( \left[\mathbb{I}_{i-1,i}(\thetaa_j)-\Delta\chi(\lambda_i)\right]\big)\big( \left[\mathbb{I}_{i-1,i}(\thetaa_l)-\Delta\chi(\lambda_i)\right]\big)\Big]& =0 \qquad \qquad \text{for }j\neq l\\
	\Exp\Big[\big( \left[\mathbb{I}_{i-1,i}(\thetaa_j)-\Delta\chi(\lambda_i)\right]\big)^2\Big]&= \mathrm{Var}\big(\mathbb{I}_{i-1,i}(\thetaa_j)\big)\\
	&=\Delta\chi(\lambda_i)\left[1-\Delta\chi(\lambda_i)\right] \\
	\end{split}
	\end{equation}
	Using these relations,
	\begin{equation}
	\begin{split}
	\mathrm{Var}(\Delta\widehat\chii)&=\frac{\Delta\chi(\lambda_i)\left[1-\Delta\chi(\lambda_i)\right]}{N}\\
	\textrm{COV}(\Delta\widehat\chii)&=\sqrt{\frac{1-\Delta\chi(\lambda_i)}{N\Delta\chi(\lambda_i)}}
	\end{split}
	\end{equation}
\end{proof}


    
    
    
    


\bibliographystyle{elsarticle-num}

\bibliography{Farzad_database_usc}

\begin{thebibliography}{10}
\expandafter\ifx\csname url\endcsname\relax
  \def\url#1{\texttt{#1}}\fi
\expandafter\ifx\csname urlprefix\endcsname\relax\def\urlprefix{URL }\fi
\expandafter\ifx\csname href\endcsname\relax
  \def\href#1#2{#2} \def\path#1{#1}\fi

\bibitem{sakamoto1986akaike}
Y.~Sakamoto, M.~Ishiguro, G.~Kitagawa, Akaike information criterion statistics,
  Dordrecht, The Netherlands: D. Reidel 81~(10.5555) (1986) 26853.

\bibitem{schwarz1978estimating}
G.~Schwarz, et~al., Estimating the dimension of a model, The Annals of
  Statistics 6~(2) (1978) 461--464.

\bibitem{Claeskens2003}
G.~Claeskens, N.~L. Hjort, The focused information criterion, Journal of the
  American Statistical Association 98~(464) (2003) 900--916.

\bibitem{Gruenwald2019}
P.~Grünwald, T.~Roos, Minimum description length revisited, Int. J. Math. Ind.
  11~(01) (2019) 1930001.

\bibitem{Spiegelhalter2014}
D.~J. Spiegelhalter, N.~G. Best, B.~P. Carlin, A.~Linde, The deviance
  information criterion: 12 years on, J. R. Stat. Soc. Ser. B. Stat. Methodol.
  76~(3) (2014) 485--493.

\bibitem{popper2005logic}
K.~Popper, The logic of scientific discovery, Routledge, 2005.

\bibitem{de2018investigation}
S.~De, P.~T. Brewick, E.~A. Johnson, S.~F. Wojtkiewicz, Investigation of model
  falsification using error and likelihood bounds with application to a
  structural system, Journal of Engineering Mechanics 144~(9) (2018) 04018078.

\bibitem{burnham2002model}
K.~P. Burnham, D.~R. Anderson, Model selection and multimodel inference: a
  practical information-theoretic approach, Springer Science \& Business Media,
  2002.

\bibitem{Agnimitra2024}
A.~Dasgupta, E.~A. Johnson, Model falsification from a {B}ayesian viewpoint
  with applications to parameter inference and model selection of dynamical
  systems, Journal of Engineering Mechanics 150~(6) (2024) 04024023.

\bibitem{de2018computationally}
S.~De, E.~A. Johnson, S.~F. Wojtkiewicz, P.~T. Brewick, Computationally
  efficient {B}ayesian model selection for locally nonlinear structural dynamic
  systems, Journal of Engineering Mechanics 144~(5) (2018) 04018022.

\bibitem{de2019hybrid}
S.~De, P.~T. Brewick, E.~A. Johnson, S.~F. Wojtkiewicz, A hybrid probabilistic
  framework for model validation with application to structural dynamics
  modeling, Mechanical Systems and Signal Processing 121 (2019) 961--980.

\bibitem{rao2001model}
C.~Rao, Y.~Wu, S.~Konishi, R.~Mukerjee, On model selection, Lecture
  Notes-Monograph Series (2001) 1--64.

\bibitem{chipman2001practical}
H.~Chipman, E.~I. George, R.~E. McCulloch, M.~Clyde, D.~P. Foster, R.~A. Stine,
  The practical implementation of {Bayesian} model selection, Lecture
  Notes-Monograph Series (2001) 65--134.

\bibitem{cremers2002stock}
K.~M. Cremers, Stock return predictability: {A Bayesian} model selection
  perspective, Review of Financial Studies 15~(4) (2002) 1223--1249.

\bibitem{Crouse2022}
W.~L. Crouse, G.~R. Keele, M.~S. Gastonguay, G.~A. Churchill, W.~Valdar, A
  {B}ayesian model selection approach to mediation analysis, PLOS Genetics
  18~(5) (2022) 1--33.

\bibitem{andrieu2001model}
C.~Andrieu, P.~Djuri{\'c}, A.~Doucet, Model selection by {MCMC} computation,
  Signal Processing 81~(1) (2001) 19--37.

\bibitem{elsheikh2014efficient}
A.~H. Elsheikh, I.~Hoteit, M.~F. Wheeler, Efficient bayesian inference of
  subsurface flow models using nested sampling and sparse polynomial chaos
  surrogates, Computer Methods in Applied Mechanics and Engineering 269 (2014)
  515--537.

\bibitem{madireddy2015bayesian}
S.~Madireddy, B.~Sista, K.~Vemaganti, A bayesian approach to selecting
  hyperelastic constitutive models of soft tissue, Computer Methods in Applied
  Mechanics and Engineering 291 (2015) 102--122.

\bibitem{asaadi2017computational}
E.~Asaadi, P.~S. Heyns, A computational framework for bayesian inference in
  plasticity models characterisation, Computer Methods in Applied Mechanics and
  Engineering 321 (2017) 455--481.

\bibitem{beck2004model}
J.~L. Beck, K.-V. Yuen, Model selection using response measurements: {Bayesian}
  probabilistic approach, Journal of Engineering Mechanics 130~(2) (2004)
  192--203.

\bibitem{beck2010bayesian}
J.~L. Beck, Bayesian system identification based on probability logic,
  Structural Control and Health Monitoring 17~(7) (2010) 825--847.

\bibitem{cheung2010calculation}
S.~H. Cheung, J.~L. Beck, Calculation of posterior probabilities for {Bayesian}
  model class assessment and averaging from posterior samples based on dynamic
  system data, Computer-Aided Civil and Infrastructure Engineering 25~(5)
  (2010) 304--321.

\bibitem{muto2008bayesian}
M.~Muto, J.~L. Beck, Bayesian updating and model class selection for hysteretic
  structural models using stochastic simulation, Journal of Vibration and
  Control 14~(1-2) (2008) 7--34.

\bibitem{mthembu2011model}
L.~Mthembu, T.~Marwala, M.~I. Friswell, S.~Adhikari, Model selection in finite
  element model updating using the {Bayesian} evidence statistic,
  Mech.~Syst.~Signal Pr. 25~(7) (2011) 2399--2412.

\bibitem{grigoriu2008solution}
M.~Grigoriu, R.~Field~Jr, A solution to the static frame validation challenge
  problem using bayesian model selection, Computer Methods in Applied Mechanics
  and Engineering 197~(29-32) (2008) 2540--2549.

\bibitem{Jaynes2003}
E.~T. Jaynes, Probability {T}heory: The {L}ogic of {S}cience, Cambridge
  University Press, Cambridge, U.K., 2003.

\bibitem{kass1995bayes}
R.~E. Kass, A.~E. Raftery, Bayes factors, J.~Amer.~Statist.~Assoc. 90~(430)
  (1995) 773--795.

\bibitem{goodman1999toward}
S.~N. Goodman, Toward evidence-based medical statistics. 2: The {B}ayes factor,
  Annals of internal medicine 130~(12) (1999) 1005--1013.

\bibitem{mackay1992bayesian}
D.~J.~C. MacKay, Bayesian interpolation, Neural Computation 4~(3) (1992)
  415--447.

\bibitem{mackay1992bayesiana}
D.~J.~C. MacKay, {Bayesian} methods for adaptive models, Ph.D. thesis,
  California Institute of Technology (1992).

\bibitem{gull1988bayesian}
S.~F. Gull, Bayesian inductive inference and maximum entropy, in:
  Maximum-entropy and Bayesian methods in science and engineering, Springer,
  1988, pp. 53--74.

\bibitem{newton1994approximate}
M.~A. Newton, A.~E. Raftery, Approximate {Bayesian} inference with the weighted
  likelihood bootstrap, Journal of the Royal Statistical Society. Series B
  (Methodological) 56~(1) (1994) 3--48.

\bibitem{raftery2006estimating}
A.~E. Raftery, M.~A. Newton, J.~M. Satagopan, P.~N. Krivitsky, Estimating the
  integrated likelihood via posterior simulation using the harmonic mean
  identity, in: J.~M. Bernardo, M.~J. Bayarri, J.~O. Berger, A.~P. Dawid,
  D.~Heckerman, A.~F.~M. Smith, M.~West (Eds.), Bayesian Statistics 8:
  Proceedings of the Eighth Valencia International Meeting, Oxford University
  Press, Oxford, 2004, pp. 1--45.

\bibitem{Tran2021}
M.-N. Tran, M.~Scharth, D.~Gunawan, R.~Kohn, S.~D. Brown, G.~E. Hawkins,
  Robustly estimating the marginal likelihood for cognitive models via
  importance sampling, Behavior Research Methods 53~(3) (2021) 1148--1165.

\bibitem{Bugallo2017}
M.~F. Bugallo, V.~Elvira, L.~Martino, D.~Luengo, J.~Miguez, P.~M. Djuric,
  Adaptive importance sampling: The past, the present, and the future, IEEE
  Signal Processing Magazine 34~(4) (2017) 60--79.

\bibitem{Schuster2018}
I.~Schuster, I.~Klebanov, Markov chain importance sampling—a highly efficient
  estimator for mcmc, Journal of Computational and Graphical Statistics 30
  (2018) 260 -- 268.

\bibitem{neal2001annealed}
R.~M. Neal, Annealed importance sampling, Stat.~Comput. 11~(2) (2001) 125--139.

\bibitem{Li2023}
Y.~Li, N.~Wang, J.~Yu, Improved marginal likelihood estimation via power
  posteriors and importance sampling, Journal of Econometrics 234~(1) (2023)
  28--52.

\bibitem{friel2008marginal}
N.~Friel, A.~N. Pettitt, Marginal likelihood estimation via power posteriors,
  J.~R.~Stat.~Soc.~Ser.~B Stat.~Methodol. 70~(3) (2008) 589--607.

\bibitem{ching2007transitional}
J.~Ching, Y.-C. Chen, Transitional {Markov} chain {Monte Carlo} method for
  {Bayesian} model updating, model class selection, and model averaging,
  J.~Eng.~Mech. 133~(7) (2007) 816--832.

\bibitem{sandhu2014bayesian}
R.~Sandhu, M.~Khalil, A.~Sarkar, D.~Poirel, Bayesian model selection for
  nonlinear aeroelastic systems using wind-tunnel data, Computer Methods in
  Applied Mechanics and Engineering 282 (2014) 161--183.

\bibitem{sandhu2017bayesian}
R.~Sandhu, C.~Pettit, M.~Khalil, D.~Poirel, A.~Sarkar, Bayesian model selection
  using automatic relevance determination for nonlinear dynamical systems,
  Computer Methods in Applied Mechanics and Engineering 320 (2017) 237--260.

\bibitem{chib2001marginal}
S.~Chib, I.~Jeliazkov, Marginal likelihood from the metropolis--hastings
  output, Journal of the American statistical association 96~(453) (2001)
  270--281.

\bibitem{botev2012efficient}
Z.~I. Botev, D.~P. Kroese, Efficient {M}onte {C}arlo simulation via the
  generalized splitting method, Statistics and Computing 22~(1) (2012) 1--16.

\bibitem{ChiachioBeckChiachioRus2014}
M.~Chiachio, J.~L. Beck, J.~Chiachio, G.~Rus, Approximate {B}ayesian
  computation by subset simulation, SIAM Journal on Scientific Computing 36~(3)
  (2014) A1339--A1358.

\bibitem{DiazDelaO_GarbunoInigo_Au_Yoshida_2017}
F.~A. DiazDelaO, A.~Garbuno-Inigo, S.~K. Au, I.~Yoshida, Bayesian updating and
  model class selection with subset simulation, Computer Methods in Applied
  Mechanics and Engineering 317 (2017) 1102 -- 1121.

\bibitem{VakilzadehHuangBeckAbrahamsson2017}
M.~K. Vakilzadeh, Y.~Huang, J.~L. Beck, T.~Abrahamsson, Approximate {B}ayesian
  computation by subset simulation using hierarchical state-space models,
  Mechanical Systems and Signal Processing 84 (2017) 2 -- 20.

\bibitem{nagel2016spectral}
J.~B. Nagel, B.~Sudret, Spectral likelihood expansions for {B}ayesian
  inference, Journal of Computational Physics 309 (2016) 267--294.

\bibitem{skilling2006nested}
J.~Skilling, Nested sampling for general {Bayesian} computation, Bayesian
  Analysis 1~(4) (2006) 833--859.

\bibitem{Feroz_2009}
F.~Feroz, M.~P. Hobson, M.~Bridges, {MultiNest}: an efficient and robust
  {Bayesian} inference tool for cosmology and particle physics, Monthly Notices
  of the Royal Astronomical Society 398~(4) (2009) 1601–1614.

\bibitem{Dittmann2024}
A.~J. Dittmann, Notes on the practical application of nested sampling:
  Multinest, (non)convergence, and rectification, arXiv preprint
  arXiv:2404.16928 (2024).

\bibitem{Feroz2013}
F.~Feroz, M.~Hobson, E.~Cameron, A.~Pettitt, Importance nested sampling and the
  {MultiNest} algorithm, The Open Journal of Astrophysics 2 (Jun. 2013).

\bibitem{polson2014vertical}
N.~G. Polson, J.~G. Scott, Vertical-likelihood {M}onte {C}arlo, arXiv preprint
  arXiv:1409.3601 (2014).

\bibitem{Llorente2023}
F.~Llorente, L.~Martino, D.~Delgado, J.~L\'{o}pez-Santiago, Marginal likelihood
  computation for model selection and hypothesis testing: An extensive review,
  SIAM Review 65~(1) (2023) 3--58.

\bibitem{shaw2007efficient}
J.~Shaw, M.~Bridges, M.~Hobson, Efficient {Bayesian} inference for multimodal
  problems in cosmology, Monthly Notices of the Royal Astronomical Society
  378~(4) (2007) 1365--1370.

\bibitem{feroz2008multimodal}
F.~Feroz, M.~P. Hobson, Multimodal nested sampling: an efficient and robust
  alternative to {Markov Chain Monte Carlo} methods for astronomical data
  analyses, Monthly Not., Royal Astro.~Soc. 384~(2) (2008) 449--463.

\bibitem{ross2013simulation}
S.~M. Ross, Simulation, 5th Edition, Academic Press, 2013.

\bibitem{mckay1979comparison}
M.~D. McKay, R.~J. Beckman, W.~J. Conover, Comparison of three methods for
  selecting values of input variables in the analysis of output from a computer
  code, Technometrics 21~(2) (1979) 239--245.

\bibitem{chib1995understanding}
S.~Chib, E.~Greenberg, Understanding the {M}etropolis-{H}astings algorithm, The
  American Statistician 49~(4) (1995) 327--335.

\bibitem{chib2001markov}
S.~Chib, Markov chain {Monte Carlo} methods: computation and inference, in:
  Handbook of Econometrics, Vol.~5, Elsevier, 2001, pp. 3569--3649.

\bibitem{andrieu2003introduction}
C.~Andrieu, N.~De~Freitas, A.~Doucet, M.~I. Jordan, An introduction to {MCMC}
  for machine learning, Machine Learning 50~(1) (2003) 5--43.

\bibitem{tierney1994markov}
L.~Tierney, Markov chains for exploring posterior distributions, The Annals of
  Statistics (1994) 1701--1728.

\bibitem{au2001estimation}
S.-K. Au, J.~L. Beck, Estimation of small failure probabilities in high
  dimensions by subset simulation, Probabilistic Engineering Mechanics 16~(4)
  (2001) 263--277.

\bibitem{mukherjee2006nested}
P.~Mukherjee, D.~Parkinson, A.~R. Liddle, A nested sampling algorithm for
  cosmological model selection, The Astrophysical Journal Letters 638~(2)
  (2006) L51--L54.

\bibitem{del2005genealogical}
P.~Del~Moral, J.~Garnier, Genealogical particle analysis of rare events, The
  Annals of Applied Probability 15~(4) (2005) 2496--2534.

\bibitem{cerou2012sequential}
F.~C{\'e}rou, P.~Del~Moral, T.~Furon, A.~Guyader, Sequential {Monte Carlo} for
  rare event estimation, Statistics and Computing 22~(3) (2012) 795--808.

\bibitem{chopin2010properties}
N.~Chopin, C.~P. Robert, Properties of nested sampling, Biometrika 97~(3)
  (2010) 741--755.

\bibitem{latuszynski2025mcmc}
K.~{\L}atuszy{\'n}ski, M.~T. Moores, T.~Stumpf-F{\'e}tizon, {MCMC} for
  multi-modal distributions, arXiv preprint arXiv:2501.05908 (2025).

\bibitem{murphy2007conjugate}
K.~P. Murphy, Conjugate {B}ayesian analysis of the {G}aussian distribution.

\bibitem{kamalzare2014computationally}
M.~Kamalzare, E.~A. Johnson, S.~F. Wojtkiewicz, Computationally efficient
  design of optimal output feedback strategies for controllable passive damping
  devices, Smart Materials and Structures 23~(5) (2014) 055027.

\bibitem{wen1976method}
Y.-K. Wen, Method for random vibration of hysteretic systems,
  J.~Eng.~Mech.~Div.~{ASCE} 102~(2) (1976) 249--263.

\bibitem{skinner1993}
R.~I. Skinner, W.~H. Robinson, G.~H. McVerry, An Introduction to Seismic
  Isolation, Wiley, Chichester, 1993.

\bibitem{subcommittee2010guide}
Guide specifications for seismic isolation design, American Association of
  State Highway and Transportation Officials, Washington D.C., 2010.

\bibitem{langtangen2016solving}
H.~P. Langtangen, A.~Logg, Solving {PDEs} in {Python}: the {FEniCS} tutorial I,
  Springer, 2016.

\bibitem{cheung2011bayesian}
S.~H. Cheung, T.~A. Oliver, E.~E. Prudencio, S.~Prudhomme, R.~D. Moser,
  Bayesian uncertainty analysis with applications to turbulence modeling,
  Reliability Engineering \& System Safety 96~(9) (2011) 1137--1149.

\bibitem{oliver2011bayesian}
T.~A. Oliver, R.~D. Moser, Bayesian uncertainty quantification applied to
  {RANS} turbulence models, in: Journal of Physics: Conference Series, Vol.
  318, IOP Publishing, 2011, p. 042032.

\bibitem{edeling2014bayesian}
W.~Edeling, P.~Cinnella, R.~P. Dwight, H.~Bijl, Bayesian estimates of parameter
  variability in the k--$\varepsilon$ turbulence model, Journal of
  Computational Physics 258 (2014) 73--94.

\bibitem{edeling2014predictive}
W.~Edeling, P.~Cinnella, R.~P. Dwight, Predictive {RANS } simulations via
  {Bayesian} model-scenario averaging, Journal of Computational Physics 275
  (2014) 65--91.

\bibitem{Handley2015}
W.~J. Handley, M.~P. Hobson, A.~N. Lasenby, polychord: nested sampling for
  cosmology, Monthly Notices of the Royal Astronomical Society: Letters 450~(1)
  (2015) L61--L65.

\bibitem{ghanem1991stochastic}
R.~G. Ghanem, P.~D. Spanos, Stochastic finite element method: Response
  statistics, in: Stochastic Finite Elements: A Spectral Approach, Springer,
  1991, pp. 101--119.

\bibitem{efron2004least}
B.~Efron, T.~Hastie, I.~Johnstone, R.~Tibshirani, Least angle regression, The
  Annals of Statistics 32~(2) (2004) 407--499.

\bibitem{feller2008introductionv1}
W.~Feller, An introduction to probability theory and its applications, Vol.~1,
  John Wiley \& Sons, 2008.

\end{thebibliography}

\end{document}